\documentclass[journal]{IEEEtran}

\usepackage{amssymb}
\usepackage{amsmath,amssymb,amsfonts}
\usepackage{multirow}
\usepackage{graphicx}
\usepackage{textcomp}
\usepackage{amsthm}
\renewcommand{\vec}[1]{\boldsymbol{#1}}
\usepackage{setspace}
\usepackage{subfigure}

\usepackage{graphicx}
\usepackage{algorithm}
\usepackage{algorithmic}
\usepackage{amsmath}
\usepackage{color}
\usepackage{amsfonts}

\ifCLASSINFOpdf
\else
\fi

\newtheorem{thm}{Theorem}
\newtheorem{lem}{Lemma}

\newtheorem{defn}{Definition}

\begin{document}

\title{Differential Privacy-Based Online Allocations towards Integrating Blockchain and Edge Computing}

\author{Jianxiong Guo,
	Weili Wu,~\IEEEmembership{Senior Member,~IEEE}
	\thanks{J. Guo and W. Wu are with the Department
		of Computer Science, Erik Jonsson School of Engineering and Computer Science, Univerity of Texas at Dallas, Richardson, TX 75080, USA. \textit{(Corresponding author: Jianxiong Guo.)}
		
		E-mail: jianxiong.guo@utdallas.edu}
	\thanks{Manuscript received April 19, 2005; revised August 26, 2015.}}

\markboth{Journal of \LaTeX\ Class Files,~Vol.~14, No.~8, August~2015}%
{Shell \MakeLowercase{\textit{et al.}}: Bare Demo of IEEEtran.cls for IEEE Journals}

\maketitle

\begin{abstract}
	In recent years, the blockchain-based Internet of Things (IoT) has been researched and applied widely, where each IoT device can act as a node in the blockchain. However, these lightweight nodes usually do not have enough computing power to complete the consensus or other computing-required tasks. Edge computing network gives a platform to provide computing power to IoT devices. A fundamental problem is how to allocate limited edge servers to IoT devices in a highly untrustworthy environment. In a fair competition environment, the allocation mechanism should be online, truthful, and privacy safe. To address these three challenges, we propose an online multi-item double auction (MIDA) mechanism, where IoT devices are buyers and edge servers are sellers. In order to achieve the truthfulness, the participants' private information is at risk of being exposed by inference attack,which may lead to malicious manipulation of the market by adversaries. Then, we improve our MIDA mechanism based on differential privacy to protect sensitive information from being leaked. It interferes with the auction results slightly but guarantees privacy protection with high confidence. Besides, we upgrade our privacy-preserving MIDA mechanism such that adapting to more complex and realistic scenarios. In the end, the effectiveness and correctness of algorithms are evaluated and verified by theoretical analysis and numerical simulations.
\end{abstract}


\IEEEpeerreviewmaketitle

\section{Introduction}
\IEEEPARstart{R}{ecently}, with the rapid development of 5G networks accompanied by higher speed and lower delay, the sensors distributed in our lives, such as mobile phones, cemera, and automobiles, can play a more important role. Therefore, Internet of Things (IoT) has become a hot topic, which connects the physical environment to the syberspare system. More and more people are paying attention to it in academia and industry. IoT can be used to achieve a variety of industrial applications, such as smart home, manufacturing, healthcare, and smart grid \cite{dai2019blockchain}. The traditional IoT paradigms are managed and coordinated by a centralized cloud server. Even though it is convenient and efficient, the centralized management mode has its inherent defects, such as single point of failure, network congestion and so on. Through blockchain technology \cite{nakamoto2008bitcoin}, its decentralized nature helps avoid single points of failure and improve security. 
Edge computing \cite{shi2016edge} \cite{luo2020edge} replaces the centralized cloud platform at the edge of the network, supports resource/delay sensitive devices, and effectively solves the problem of network congestion.
Thus, integrating blockchain and edge computing has become the development trend of the next generation IoT.

The emergence of blockchain technology brought new opportunities to the development of IoT. Blockchain \cite{nakamoto2008bitcoin} is a public and decentralized database that is used to store real-time data generated by all valid participants in the system without a third-party platform. Because they do not trust each other, each new generated data should be verified in a distributed manner before being added into a block, and then each new generated block should be validated as well by the consensus process before being added into the blockchain in a permanent and tamper-resistant manner. At the same time, the integrity of blocks in the blockchain can be guaranteed by some crytographic methods, such as asymmetric encryption algorithms and digital signatures \cite{johnson2001elliptic}. Moreover, each piece of data in the blockchain is traceable to all valid participants because of its chain-based structure. The security and reliability of blockchain-based systems are from their consensus mechanisms. Taking Proof-of-Works (PoW) \cite{nakamoto2008bitcoin} \cite{wood2014ethereum} as an example, each participant is required to solve a hash puzzle when competing with others to generate a new block, which is very computing-intensive and time-consuming.

For a blockchain-based IoT system consisting of many lightweight IoT devices, it is difficult for the devices to participate in the consensus process due to lack of sufficient computing power. However, by integrating edge computing, edge servers at the network edge can provide computing power to their neighboring IoT devices, which is helpful for building a fully functional IoT system. \cite{pan2018edgechain} \cite{yao2019resource} \cite{chang2020incentive} \cite{ding2020incentive} \cite{ding2020pricing}. Due to the reward from a blockchain-based system for participating in its consensus process, the IoT devices are willing to consume computing power for competing with others. In addition, these IoT devices may have their own computing tasks. For example, they need to train their deep learning model according to real-time monitoring data. Whether participating in consensus or completing computing tasks, the lightweight IoT device needs to purchase computing power from the edge servers nearby. Therefore, the fundamental problem in this paper, how to allocate limited edge servers to IoT devices in such a highly untrustworthy environment, is formulated. 

However, there are several challenges that are not concerned in existing mechanism designs \cite{pan2018edgechain} \cite{yao2019resource} \cite{chang2020incentive} \cite{ding2020incentive} \cite{ding2020pricing}. In order to ensure the fairness and security of trading between IoT devices and edge servers, the allocation mechanism should be online, truthful, and privacy safe. In this paper, we not only put forward a new allocation model, but also our designs are based on addressing these three issues.

Because the consensus process in a blockchain-based system is executed round by round and edge servers runs in a dynamic environment, the allocation mechanism must be online, which is executed round by round as well. In each round, the state of each IoT devices or edge server is different. Besides, truthfulness should be guaranteed so as to prevent malicious users from manipulating the trading market. This not only ensures the fairness of trading, but also ensures the security of the blockchain-based system. Based on these two aspects, we propose an online seal-bid multi-item double auction (MIDA) mechanism to achieve the resource allocation between IoT devices and edge servers. In the system, the IoT devices are buyers and edge servers are sellers. The MIDA mechanism can give us a one-to-one mapping between IoT devices and edge servers, which is used as a special case to conduct theoretical analysis. We show that it is individually rational, budget balanced, computationally efficient, and truthful. Then, we upgrade the MIDA to MIDA-G mechanism to establish a many-to-one mapping between IoT devices and edge servers, which is more realistic.

In such an auction scenario, it exists possible risk of privacy leakage in the multiple round of truthful bids given by buyers (truthful asks given by sellers), where the bids (asks) are private information of buyers (sellers). As we know, the auction results could be altered by the change in a single bid (ask) \cite{jin2016enabling} \cite{chen2019differentially}. Thus, our MIDA mechanism is vulnerable to inference attack \cite{zhu2014differentially} \cite{zhu2015differentially}, where adversaries could infer bids (asks) of other buyers (sellers) through comparing the auction results of multiple rounds. Thus, the adversary can make the auction result beneficial to itself by manipulating its own bid or ask. The differential privacy \cite{dwork2008differential} is a promising technology to prevent the adversary from inferring other truthful bids or asks through the public auction results. Commonly used schemes of achieving differential privacy with theoretical guarantee include exponetial mechanism and Laplace mechanism. Thus, another important part of this paper is to design a scheme with differential private protection for our MIDA mechanism so as to achieve privacy protection. Our main contributions in this paper can be summarized as follows.

\begin{itemize}
	\item We propose an online MIDA and MIDA-G mechanism to model the allocation of computing power between IoT devices and edge servers.
	\item By theoretical analysis, we show that our auction mechanism is individually rational, budget balanced, computationally efficient, and truthful.
	\item We give an example to demonstrate the risk of potential privacy leakage from inference attack, and propose a differential privacy-based MIDA-AP and MIDA-G-AP mechanism based on Laplace mechanism. This can effectively ensure privacy protection without affecting the truthfulness of our auction mechanisms.
\end{itemize}

\textbf{Orgnizations: }In Sec. \uppercase\expandafter{\romannumeral2}, we discuss the-state-of-art work. In Sec. \uppercase\expandafter{\romannumeral3}, we introduce the system model and define our problem formally. In Sec. \uppercase\expandafter{\romannumeral4}, we present our MIDA mechanism and theoretical analysis elaborately. In Sec. \uppercase\expandafter{\romannumeral5}, we achieve the differetial private strategy for our MIDA mechanism. Then, the more general MIDA-G mechanism is shown in Sec.\uppercase\expandafter{\romannumeral6}. Finally, we evaluate our algorithms by numerical simulations in Sec.\uppercase\expandafter{\romannumeral7} and show the conclusions in Sec. \uppercase\expandafter{\romannumeral8}.

\section{Related Work}
With the increasing development of blockchains, research on the blockchain-based IoT has attracted more and more attention. They exploited the decentralization of blockchain to achieve security, interoperability, privacy, and traceability \cite{dai2019blockchain} \cite{wu2019comprehensive}. For the resource allocation between IoT devices and edge servers, we summarize several classic articles here. Yao \textit{et al.} \cite{yao2019resource} studied the resource management and pricing problem between miners and cloud servers by Stackelberg game and reinforcement learning algorithm. Chang \textit{et al.} \cite{chang2020incentive} investigated how to encourage miners to purchase the computing resources from edge service provider and find the optimal solution by a two-stage Stackelberg game. Ding \textit{et al.} \cite{ding2020incentive} \cite{ding2020pricing} attempted to build a secure blockchain-based IoT system by attracting more IoT devices to purchase computing power from edge servers and participate in the consensus process, where they adopted a multi-leader multi-follower Stackelberg game. However, the truthfulness cannot be guaranteed by the above methods, which is hard to ensure the fairness.

Auction theory has been considered as a feasible solution in many different systems, such as mobile crowdsensing \cite{yang2015incentive} \cite{guo2020reliable}, mobile cloud/edge computing \cite{jin2015auction} \cite{jin2015auction1}, and energy trading \cite{wang2012designing} \cite{yassine2019double}. Here, we only focus on the multi-item double auction. Yang \textit{et al.} \cite{yang2011truthful} studied the cooperative communications by proposing a double auction mechanism, where they first got a mapping from buyer to seller through assignment algorithms and then used McAfree auction \cite{mcafee1992dominant} to determine winners and clearing prices. Jin \textit{et al.} \cite{jin2015auction} \cite{jin2015auction1} considered a resource allocation problem by designing a truthful double auction mechanism for the resource trading between users and cloudlet, but it is only one-to-one mapping. Guo \textit{et al.} \cite{guo2020double} proposed a secure and efficient charging scheduling system based on DAG-blockchain and double auction mechanism. Due to the online execution, double sides, multi-item, and resource-constrained edge servers, these existing methods are not suitable to our model.

In a fair auction platform, the auction results have to be public, which leads to the sensitive information of participants is at risk of being exposed. The theories and applications of privacy protection were investigated in \cite{yang2017survey} \cite{ferrag2018systematic}. Dwork \textit{et al.} \cite{dwork2008differential} first put forward the concept of differential privacy, and then it was applied to double auction mechanism for protecting players' privacy. Chen \textit{et al.} \cite{chen2019differentially} combined the differential privacy with double spectrum auction design in order to maximize social welfare approximately. Li \textit{et al.} \cite{li2019towards} proposed an online double auction scheme and combined with differential privacy to build a secure market among electric vehicles. Besides, a variety of differentially privacy-based schemes have been used to design double auction systems in \cite{hassan2019deal} \cite{li2020online} \cite{ni2020differentially}. Due to the difference of auction mechanisms, the above differential privacy-based schemes can not be directly applied to our model. We need to re-design according to the structure of our own algorithm. Actually, our differential privacy-based scheme is more simple and effective than the above works.

\section{Models and Preliminaries}
In this section, we introduce the model of integrating blockchain with edge computing, adversary attack, and definitions of double auction and differential privacy.

\subsection{System model}
Fig. \ref{fig1} illustrates the basic usage scenario that we consider in this paper. A typical instance of the blockchain-based IoT system is a network composed of many lightweight IoT devices that are used to perform some tasks such as environmental monitoring, where each IoT device can be considered as a node in the blockchain system. We assume that this blockchain system adopts a proof-of-work (PoW) consensus mechanism. Stimulated by the reward from participating in the consensus process of blockchain system, part of IoT devices would like to be miners. Then, they would compete with other miners for the right to generate the next block by solving a hash puzzle, which has been adopted widely in the Bitcoin system \cite{nakamoto2008bitcoin}. The block in the blockchain system consists of transactions stored in a merkle tree structure and block header that contains the hash value of its previous block. In the PoW consensus mechanism, the mining process is executed by miners to find a nonce such that
\begin{equation}
	Hash(Tranactions, Header, Nonce)\leq D
\end{equation}
where $D$ is a 256-bits binary number assigned by the platform to control the rate of block generation.

However, such a blockchain is limited by its high requirement for computing power because its consensus mechanism is based on solving a hash puzzle, thereby it cannot be applied to lightweight devices with limited power directly. At this time, these devices will attempt to purchase computing power from one or more edge servers around them and offload their mining tasks to their assigned edge servers. The number of edge servers is limited and each device may only offload its tasks to a few edge servers nearby. Thus, a natural question is how to allocate limited edge severs to the devices in the blockchain system. 

\begin{figure}[!t]
	\centering
	\includegraphics[width=\linewidth]{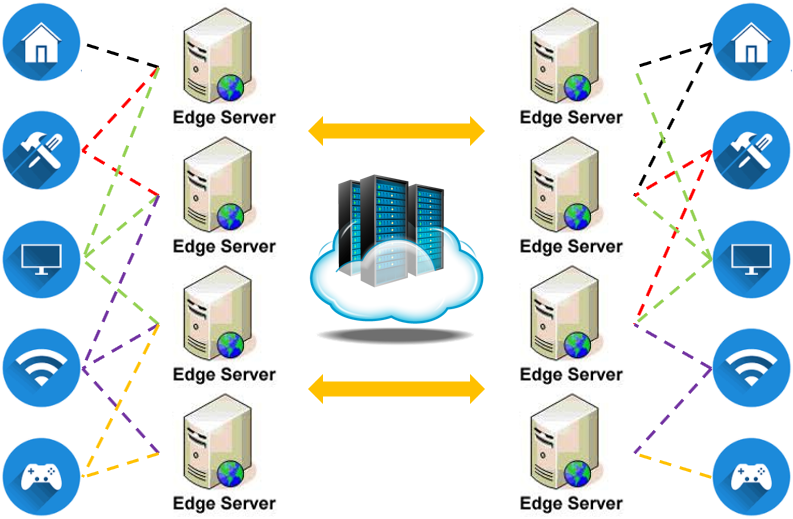}
	\centering
	\caption{The basic usage scenario of integrating blockchain and edge computing (Blockchain-based IoT system).}
	\label{fig1}
\end{figure}

\subsection{Problem Formulation}
Shown as Fig. \ref{fig1}, the dash lines that connect a device to several edge servers indicate the mining task of this device can only be offloaded to one of the edge servers that it connects to because of space location and bandwidth limitations. Therefore, we propose an online sealed-bid multi-item double auction mechanism to model the competition among devices in the blockchain system and edge servers. There are three kinds of players involved in such an auction: auctioneer, buyers, and sellers. In the system, the centralized cloud server acts as the auctioneer, devices that plan to purchase computing power act as buyers, and edge servers that can provide computing power act as sellers. The transactions between IoT devices and edge servers are established on the wireless infrastructure, which is the reason why the tasks of a device may only be assigned to a finite number of edge servers near it. Otherwise, the communication overhead will increase dramatically. Even though that, for each device, it has different preferences for different edge servers in line with the quality of service they provide, such as geographical locations, response speeds, credit scores, and other factors. This causes a device to have different valuations for different edge servers.

In the system, we denote by the set of devices (buyers) $D=\{d_1,d_2,\cdots,d_n\}$ who pay for purchase computing power and the set of edge servers (sellers) $S=\{s_1,s_2,\cdots,s_m\}$ who get reward for providing computing power. It formulates a bipartite graph $G=(D,S,E)$ like Fig. \ref{fig1} where $\{d_i,s_j\}\in E$ if device $d_i$ can be assigned to edge server $s_i$. As we know, the PoW consensus mechanism is executed round by round, where each round happens in a time slot. We consider a time interval $I$, which can be discretized into time slots as $I=\{1,2,\cdots,T\}$. In each slot $t\in I$, a round of consensus will be carried out, in other words, a round of auction will be carried out. Let us consider a special case first in which each edge server can serve at most one device in a slot. Therefore, for each slot $t\in I$, it is similar to finding a maximum weight matching in bipartite graph $G$ according to the bids of buyers and asks of sellers.

In a time slot $t\in I$, we define the notations in our auction as follows. For each buyer $d_i\in D$, its bid information can be given by $\mathcal{B}^t_i=(\vec{b}_i^t,r_i^t)$. Here, $\vec{b}_i^t=(b_{i,1}^t,b_{i,2}^t,\cdots,b_{i,m}^t)$ is its bid vector where $b_{i,j}^t\in\vec{b}_i^t$ ($b_{i,j}^t\in[b_{min},b_{max}]$) is the unit bid (maximum buying price) per unit computing power of purchasing from the seller $s_j\in S$ and $r_i^t$ is the amount of computing power it needs to buy. According to our scenario, we have $b_{i,j}^t=0$ if $\{d_i,s_j\}\notin E$. For each seller $s_j\in S$, its ask information can be given by $\mathcal{A}^t_j=(a_j^t,q_j^t)$. Here, $a_j^t\in[a_{min},a_{max}]$ is the unit ask (minimum selling price) per unit computing power and $q_j^t$ is the maximum amount of computing power it can provide. Based on the above definitions, a buyer gives different unit bids to different sellers, but a seller gives the same ask to different buyers since it only cares about the payment charged from buyers. The bid information of buyers and ask information of sellers are submitted to the auctioneer, thereby our auction can be defined on $\Omega^t=(\{\mathcal{B}_i^t\}_{d_i\in D},\{\mathcal{A}_j^t\}_{s_j\in S})$.

Given the $\Omega^t$ in the slot $t$, the auctioneer not only determines a winning buyer set $D_w^t\subseteq D$ and a winning seller set $S_w^t\subseteq S$, but also determines a bijective function $\sigma^t$ mapping from $D_w^t$ to $S_w^t$, where $\sigma^t(d_i)=s_j$ implies the computing task of device $d_i\in D_w^t$ is offloaded to edge server $s_j\in S_w^t$ in the slot $t$. Here, we can say ``device $d_i$ is assigned to edge server $s_j$''. We define the unit price $\hat{p}^t_i$ charged to buyer $d_i\in D_w^t$ and the unit payment $\bar{p}^t_j$ rewarded to seller $s_j\in S_w^t$. Besides, the valuation vector of each buyer $d_i\in D$ is $\vec{v}_i^t=(v_{i,1}^t,v_{i,2}^t,\cdots,v_{i,m}^t)$ where $v_{i,j}^t\in\vec{v}_i^t$ is the unit valuation per unit computing power of purchasing from the seller $s_j\in S$ and the unit cost per computing power of each seller $s_j\in S$ is $c_j^t$. According to the valuation vectors of buyers and costs of sellers, the utility $\hat{u}^t_i$ of buyer $d_i\in D$ and utility $\bar{u}^t_j$ of seller $s_j\in S$ can be defined. For each winning buyer $d_i\in D_w^t$ and winning seller $s_j\in S_w^t$, we have
\begin{flalign}
	&\hat{u}^t_i(\Omega^t)=\left(v^t_{i,\sigma^t(d_i)}-\hat{p}^t_i\right)\cdot r_i^t\\
	&\bar{u}^t_j(\Omega^t)=\left(\bar{p}_j^t-c_j^t\right)\cdot r_{\sigma^t_{-1}(s_j)}^t
\end{flalign}
where $\sigma^t_{-1}(s_j)=d_i$ means that $\sigma^t(d_i)=s_j$. Otherwise, we have $\hat{u}^t_i(\Omega^t)=0$ for each losing buyer $d_i\in D\backslash D_w^t$ and $\bar{u}^t_j(\Omega^t)=0$ for each losing seller $s_j\in S\backslash S_w^t$. The utilitis are characterized by the difference between charge (reward) and valuation (cost), which reflects their satisfaction of the current auction result.

For convenience, the allocation result determined by the auctioneer in the slot $t$ can be denoted by a matrix $\vec{X}^t=[x^t_{i,j}]_{d_i\in D,s_j\in S}$, where we have
\begin{equation}
	x^t_{i,j}=\left
	\{\begin{IEEEeqnarraybox}[\relax][c]{l's}
		1,&if buyer $d_i$ is assigned to seller $s_i$\\
		0,&otherwise
	\end{IEEEeqnarraybox}
	\right.
\end{equation}
Here, we consider the auctioneer as a non-profit platform whose objective is to maximum the socal welfare in current time slot. In the time slot $t\in I$, the social welfare maximization (SWM) problem is formulated:
\begin{align}
	\max\quad & \sum_{d_i\in D}\sum_{s_j\in S} x^t_{i,j}\cdot r_i^t\cdot\left(b_{i,j}^t-a_j^t\right)\label{eq5}\\
		s.\;t.\quad & \sum_{d_i\in D}x^t_{i,j}\leq 1, \;\forall s_j\in S \tag{\ref{eq5}{a}}\\
		& \sum_{s_j\in S}x^t_{i,j}\leq 1, \;\forall d_i\in D \tag{\ref{eq5}{b}}\\
		& \sum_{d_i\in D}r_i^t\cdot x^t_{i,j}\leq q_j^t, \;\forall s_j\in S \tag{\ref{eq5}{c}}\\
		& \sum_{t'\leq t}\sum_{s_j\in S}r_i^{t'}\cdot x^{t'}_{i,j}\leq\theta(D,S), \;\forall d_i\in D \tag{\ref{eq5}{d}}\\
		& x^t_{i,j}=0, \;\forall \{d_i,s_j\}\notin E \tag{\ref{eq5}{e}}\\
		& x^t_{i,j}\in\{0,1\}, \;\forall d_i\in D,\forall s_j\in S \tag{\ref{eq5}{f}}
\end{align}
Shown as (\ref{eq5}), the SWM problem illustrates all the constraints we give in our online double auction mechanism. (5a) and (5b) express the bijection relationship between winning buyer set $D_w^t$ and winning seller set $S_w^t$; (5c) implies the amount of computing power that $d_i$ buy from $s_j$ must be less than the amount of computing power that $s_j$ can provide if buyer $d_i$ is assigned to seller $s_j$; (5d) shows that the total computing resourse that buyer $d_i$ purchases in the time interval $I$ cannot exceed a threshold $\theta(D,S)$, which depends on the number of devices and edge servers in the system. It aims to prevent a device from owning too much computing power to undermine the security of blockchain system; and (5e) means that a device can only be assigned to its permitted edge servers. Actually, maximizing social welfare is just an idealized situation, and we usually need to sacrifice part of social welfare to ensure the truthfulness of auction mechanism.

\subsection{Potential Information Leakage}
For a justifiable auction platform, the auctioneer should announce the auction result in each slot $t$. The auction result refers to $\vec{X}^t$ (including $D_w^t$, $S_w^t$, and $\sigma^t$), the amount of computing power $\{r_i^t\}_{d_i\in D_w^t}$ and $\{q_j^t\}_{s_j\in S_w^t}$, and clearing unit price $\{\hat{p}^t_i\}_{d_i\in D_w^t}$ and $\{\bar{p}^t_j\}_{s_j\in S_w^t}$ to all players so as to make sure the fairness and verifiability of this auction. The reason to publish the amount of computing power is to let other players check the Constraint (5c) and (5d). Like this, the players avoid being cheated by the auctioneer because it cannot fabricate auction results for benefits. However, adversaries can use these public auction results to infer other players' private information and lead to privacy leakage. 

In our proposed auction mechanism, private information contains the unit bids of buyers $\{\vec{b}_i^t\}_{d_i\in D}$ and unit asks of sellers $\{a_j^t\}_{s_j\in S}$. Even though adversaries cannot get in touch with other players' private information, they can make inferences from these known auction results, which is called inference attack \cite{zhu2014differentially} \cite{zhu2015differentially}. Thus, we consider two kinds of privacy preservation (inference attack) in our mechanism design: (1) The adversary (some seller) infers the unit asks of other sellers; and (2) The adversary (some buyer) infers the unit bids of other buyers. This is because of the competitive relationship among buyers (or sellers). The inference attack will not only lead to the privacy leakage of players, but also make the auction unfair. We assume that all players have known how the auction operates in the beginning. By inferring other players' bids or asks, the adversary can change its strategy to increase its benefit strategically. Therefore, protecting privacy is a challenge that we must face in mechanism design.

\section{Online Double Auction Design}
In this section, we introduce several design rationales, a truthful double auction mechanism, and an example to explain the reason why private information is leaked.

\subsection{Design Rationales}
The online double auction in the time slot $t\in I$ has been defined as $\Omega^t$. A desired double auction mechanism should satisfy individual rationality, budget balance, computational efficiency, and truthfulness.
\begin{defn}[Individual Rationality]
	The utility for each player should be larger than or equal to zero. In our auction $\Omega^t$, we have $\hat{u}_i^t(\Omega^t)\geq 0$ for each buyer $d_i\in D$ and $\bar{u}^t_j(\Omega^t)\geq 0$ for each seller $s_j\in S$.
\end{defn}
\begin{defn}[Budget Balance]
	The auctioneer should be profitable to operate this auction. Thus, we have
	\begin{equation}\label{eq6}
		\sum_{d_i\in D}\hat{p}_i^t\cdot r_i^t-\sum_{s_j\in S}\bar{p}_j^t\sum_{d_i\in D} r_i^t\cdot x_{i,j}^t\geq 0
	\end{equation}
\end{defn}
\begin{defn}[Computational Efficiency]
	The auction results defined on the Sec. \uppercase\expandafter{\romannumeral3}.C can be obtained in polynomial time by the auction mechanism.
\end{defn}
\begin{defn}[Truthfulness]
	Every buyer (seller) bids (asks) truthfully is one of its dominant strategies, which maximizes its utility definitely. Thus, we have $\hat{u}^t_i((\vec{v}_i^t,r_i^t),\Omega^t_{-i})\geq\hat{u}^t_i((\vec{b}_i^t,r_i^t),\Omega^t_{-i})$ for each buyer $d_i\in D$ and $\bar{u}^t_j((c_j^t,q_j^t),\Omega^t_{-j})\geq\bar{u}^t_j((a_j^t,q_j^t),\Omega^t_{-j})$ for each seller $s_j\in S$, where $\Omega^t_{-i}$ ($\Omega^t_{-j}$) is the strategy collection of players execpt buyer $d_i$ (seller $s_j$). If an auction is truthful, there is no buyer improving its utility by giving a bid vector different from its valuation vector and no seller improving its utility by giving a ask different from its cost.
\end{defn}

When we consider the truthfulness, we suppose that the amount of computing power $r_i^t$ ($g_j^t$) submitted by the buyer (seller) is authentic and public since it can be monitored and must be executed once assigned. Becasue of the truthfulness, no player has the motivation to change its strategy for obtaining more benefit, which makes the strategic decision of players easier and guarantees a fair competitive environment.

\subsection{Algorithm Design}
Here, we propose a multi-item double auction (MIDA) mechanism that attempts to maximize the social welfare but ensure the truthfulness. It is shown in Algorithm \ref{a1}, which consists of two parts, winning candidate determination (MIDA-WCD) shown in Algorithm \ref{a2} and assignment \& pricing (MIDA-AP) shown in Algorithm \ref{a3}.

\begin{algorithm}[!t]
	\caption{\text{MIDA $(\Omega^t)$}}\label{a1}
	\begin{algorithmic}[1]
		\renewcommand{\algorithmicrequire}{\textbf{Input:}}
		\renewcommand{\algorithmicensure}{\textbf{Output:}}
		\REQUIRE $\Omega^t=(\{\mathcal{B}_i^t\}_{d_i\in D},\{\mathcal{A}_j^t\}_{s_j\in S})$
		\ENSURE $D_w^t,S_w^t,\sigma^t,\hat{P}_w^t,\bar{P}_w^t$
		\STATE $(D_c^t,S_c^t,a_{j_\phi}^t)\leftarrow$ MIDA-WCD $(\Omega^t)$
		\STATE $(D_w^t,S_w^t,\sigma^t,\hat{P}_w^t,\bar{P}_w^t)\leftarrow$ MIDA-AP $(\Omega^t,D_c^t,S_c^t,a_{j_\phi}^t)$
		\RETURN $(D_w^t,S_w^t,\sigma^t,\hat{P}_w^t,\bar{P}_w^t)$ 
	\end{algorithmic}
\end{algorithm}

\begin{algorithm}[!t]
	\caption{\text{MIDA-WCD $(\Omega^t)$}}\label{a2}
	\begin{algorithmic}[1]
		\renewcommand{\algorithmicrequire}{\textbf{Input:}}
		\renewcommand{\algorithmicensure}{\textbf{Output:}}
		\REQUIRE $\Omega^t=(\{\mathcal{B}_i^t\}_{d_i\in D},\{\mathcal{A}_j^t\}_{s_j\in S})$
		\ENSURE $D_c^t,S_c^t,a_{j_\phi}^t$
		\STATE $D_c^t\leftarrow\emptyset,S_c^t\leftarrow\emptyset$
		\STATE Construct a set $D^t_*=\{d^t_{k,l}:d^t_{k,l}\text{ satisfies (7) (8)}\}$
		\STATE Sort the sellers such that $S^t_*=\langle s_{j_1}^t,s_{j_2}^t,\cdots,s_{j_m}^t\rangle$ where $a_{j_1}^t\leq a_{j_2}^t\leq\cdots\leq a_{j_m}^t$	
		\STATE Find the median ask $a_{j_\phi}^t$ of $S^t_*$, $\phi=\left\lceil\frac{m+1}{2}\right\rceil$
		\FOR {each $d^t_{k,l}\in D^t_*$}
		\IF {$b^t_{k,l}\geq a_{j_\phi}^t$ and $a_l^t<a_{j_\phi}^t$}
		\STATE $D_c^t\leftarrow D_c^t\cup\{d^t_{k,l}\}$
		\IF {$s_l\notin S_c^t$}
		\STATE $S_c^t\leftarrow S_c^t\cup\{s_l\}$
		\ENDIF
		\ENDIF
		\ENDFOR
		\RETURN $(D_c^t,S_c^t,a_{j_\phi}^t)$
	\end{algorithmic}
\end{algorithm}

Shown as Algorithm \ref{a2}, we construct a set of buyer-seller pairs $D^t_*$ first where each pair $d_{k,l}^t\in D^t_*$ if
\begin{flalign}
	& d_k\in D,s_l\in S,\{d_k,s_l\}\in E,b^t_{k,l}>0,r_k^t\leq q_l^t\\
	& \sum_{t'<t}\sum_{s_j\in S}r_k^{t'}\cdot x^{t'}_{k,j}+r_k^t\leq\theta(D,S)
\end{flalign}
It can be denoted by $D^t_*=\{d^t_{k,l}:d^t_{k,l}\text{ satisfies (7) (8)}\}$, which means that the buyer $d_k$ is feasible to be assigned to the seller $s_l$ in the slot $t$. Then, we sort the sellers based on their asks in an ascending order and select the median $a_{j_\phi}^t$ as a threshold to balance the number of winning buyer candidates and winning seller candidates. For each pair $d^t_{k,l}\in D^t_*$, it will be a winning buyer candidate if its bid $b_{k,l}^t$ is not less than $a_{j_\phi}^t$ and the ask of its corresponding seller $a_l^t$ is less than $a_{j_\phi}^t$. At the same time, seller $s_l$ will be a winning seller candidate if there is at least one winning buyer candidate $d^t_{k,l}\in D_c^t$ existing that bids for it.  

\begin{algorithm}[!t]
	\caption{\text{MIDA-AP $(\Omega^t,D_c^t,S_c^t,a_{j_\phi}^t)$}}\label{a3}
	\begin{algorithmic}[1]
		\renewcommand{\algorithmicrequire}{\textbf{Input:}}
		\renewcommand{\algorithmicensure}{\textbf{Output:}}
		\REQUIRE $\Omega^t,D_c^t,S_c^t,a_{j_\phi}^t$
		\ENSURE $D_w^t,S_w^t,\sigma^t,\hat{P}_w^t,\bar{P}_w^t$
		\STATE $D_w^t\leftarrow\emptyset,S_w^t\leftarrow\emptyset,\hat{P}_w^t\leftarrow\emptyset,\bar{P}_w^t\leftarrow\emptyset$
		\STATE Create a sorted list $Q_j^t=\langle d^t_{k,j}:d^t_{k,j}\in D_c^t\rangle$ for each $s_j\in S_c^t$ such that $b_{k_1,j}^t\cdot r_{k_1}\geq b_{k_2,j}^t\cdot r_{k_2}\geq\cdots$
		\FOR {each $s_l\in S_c^t$}
		\STATE $d_{k_1,l}^t\leftarrow Q_l^t[1]$ // \textit{The fisrt pair in $Q_l^t$}
		\IF {$d_{k_1}\notin D_w^t$}
		\STATE $D_w^t\leftarrow D_w^t\cup\{d_{k_1}\}$
		\ENDIF
		\IF {$|Q_l^t|==1$}
		\STATE $\hat{p}^t_{k_1,l}\leftarrow a_{j_\phi}^t$
		\ELSE
		\STATE $d_{k_2,l}^t\leftarrow Q_l^t[2]$
		\STATE $\hat{p}^t_{k_1,l}\leftarrow\max\{a_{j_\phi}^t, b_{k_2,l}^t\cdot(r_{k_2}^t/r_{k_1}^t)\}$
		\ENDIF
		\ENDFOR
		\FOR {each $d_k\in D_w^t$}
		\STATE $H^t_k=\{s_l:s_l\in S_c^t,Q_l^t[1]=d^t_{k,l}\}$
		\STATE Find $s_{l'}\leftarrow\arg\max_{s_l\in H_k^t}\{(b_{k,l}^t-\hat{p}_{k,l}^t)\cdot r_k^t\}$
		\STATE $\sigma^t(d_k)\leftarrow s_{l'}$
		\STATE $S_w^t\leftarrow S_w^t\cup\{s_{l'}\}$
		\STATE $\hat{p}^t_k\leftarrow\hat{p}^t_{k,l'}$, $\bar{p}^t_{l'}\leftarrow a_{j_\phi}^t$
		\STATE $\hat{P}_w^t\leftarrow\hat{P}_w^t\cup\{\hat{p}^t_k\}$, $\bar{P}_w^t\leftarrow\bar{P}_w^t\cup\{\bar{p}^t_{l'}\}$
		\ENDFOR
		\RETURN $(D_w^t,S_w^t,\sigma^t,\hat{P}_w^t,\bar{P}_w^t)$
	\end{algorithmic}
\end{algorithm}

Shown as Algorithm \ref{a3}, we create a sorted list $Q_j^t$ for each winning seller candidate $s_j\in S_c^t$ that contains all winning buyer candidates $\langle d^t_{k,j}:d^t_{k,j}\in D_c^t\rangle$ bidding for it and is sorted according to the total bid $d^t_{k,j}\cdot r_k^t$. The total bid is equal to the unit bid multiply by the amount of computing power. From line 3 to 14 in Algorithm \ref{a3}, it determines the target (buyer) of providing service for each winning seller candidate and the corresponding unit price charged to the target. For each $s_l\in S_c^t$, its target is buyer $d_{k_1}$ where the $d_{k_1,l}^t$ is the first part in $Q_l^t$ and $\hat{p}_{k_1,l}^t$ is the unit price charged to buyer $d_{k_1}$ if the $d_{k_1}$ will be assigned to $s_l$ next. From line 15 to 22 in Algorithm \ref{a3}, for each winnig buyer $d_k\in D_w^t$, it is assigned to the seller $s_{l'}$ that can obtain its maximum utility, thus we have $b_{k,l}^t-\hat{p}_{k,l}^t\geq b_{k,l'}^t-\hat{p}_{k,l'}^t$ for each $s_l\in H_k^t$. Then, the $s_{l'}$ is selected as a winning buyer, and we have $\hat{p}^t_k=\hat{p}^t_{k,l'}$ as well as $\bar{p}^t_{l'}=a_{j_\phi}^t$.

\subsection{Theoretical Analysis of MIDA}
Next, we show that our MIDA shown as Algorithm \ref{a1} satisfies above design rationales.
\begin{lem}\label{lem1}
	The MIDA is individually rational.
\end{lem}
\begin{proof}
	For each winning buyer $d_k\in D_w^t$, the unit price $\hat{p}^t_{k}$ charged to it is either $a_{j_\phi}^t$ or $b_{k_2,\sigma^t(k)}^t\cdot(r_{k_2}^t/r_{k}^t)$ where we have $d_{k_2,\sigma^t(k)}^t=Q^t_{\sigma^t(k)}[2]$. We have known that $b^t_{k,\sigma^t(k)}\geq a_{j_\phi}^t$ or $b^t_{k,\sigma^t(k)}\cdot r_k^t\geq b^t_{k_2,\sigma^t(k)}\cdot r_{k_2}^t$ since $k=Q_{\sigma^t(k)}^t[1]$. Thus, we have $b^t_{k,\sigma^t(k)}\geq \hat{p}^t_{k}$ and $\hat{u}_k^t(\Omega^t)\geq 0$. For each winning seller $s_l\in S_w^t$, we have $\bar{p}^t_l=a_{j_\phi}^t> a^t_l$ and $\bar{u}^t_l(\Omega^t)\geq 0$. Thus, the MIDA is individually rational since both buyers and sellers are indifually rational.
\end{proof}

\begin{lem}
	The MIDA is budget balanced.
\end{lem}
\begin{proof}
	These is a bijection between winning buyer set $D_w^t$ and winning seller set $S_w^t$. Thus, Equ. (\ref{eq6}) can be written as
	\begin{equation}\label{eq9}
		\sum_{s_k\in D_w^t}(\hat{p}^t_k-\bar{p}^t_{\sigma^t(k)})\cdot r_k^t\geq 0
	\end{equation}
	since we have $\hat{p}^t_k\geq a_{j_\phi}^t$ and $\bar{p}^t_{\sigma^t(k)}<a_{j_\phi}^t$ according to Lemma \ref{lem1}. Thus, the MIDA is budget balanced.
\end{proof}

\begin{lem}
	The MIDA is computationally efficient.
\end{lem}
\begin{proof}
	Let us look at Algorithm \ref{a2} first. Sorting the sellers takes $O(m\log(m))$ time and there are at most $|E|$ pairs in $D^t_*$. From this, the running time of Algorithm \ref{a2} is bounded by $O(m\log(m)+|E|)$. Then, let us look at Algorithm \ref{a3}, it takes at most $O(\phi\cdot n\log(n))$ time to construct the $Q_j^t$ for each seller candidate $s_j\in S_c^t$. Other steps in Algorithm \ref{a3} can be complete by a constant time. The running time of Algorithm \ref{a3} is bounded by $O(\phi\cdot n\log(n))$. Therefore, it is obvious that the MIDA can be completed in polynomial time, which is computationally efficient.
\end{proof}

\begin{lem}
	The MIDA is truthful.
\end{lem}
\begin{proof}
	For each buyer $d_k\in D$, we need to judge whether  $\hat{u}^t_k((\vec{v}_k^t,r_k^t),\Omega^t_{-k})\geq\hat{u}^t_k((\vec{b}_k^t,r_k^t),\Omega^t_{-k})$.
	
	\noindent
	\textbf{(1)} The $d_k\in D_w^t$ if it bids truthfully: Based on Algorithm \ref{a3}, we have $\sigma^t(d_k)=\arg\max_{s_l\in H^t_k}\{(b_{k,l}^t-\hat{p}_{k,l}^t)\cdot r_k^t\}$. For each seller $s_l\in H^t_k$, we denoted by $\hat{u}^t_{k,l}=(v_{k,l}^t-\hat{p}_{k,l}^t)\cdot r_k^t$; Otherwise, $\hat{u}^t_{k,l}=0$. For convenience, any notation $z^t_k$ and $\underline{z}^t_k$ refer to the concepts given by truthful bid $\vec{v}^t_k$ and untruthful bid $\vec{b}^t_k$. For each seller $s_l\in H^t_k$, we consider the following two sub-cases:
	\begin{itemize}
		\item (1.a) $b^t_{k,l}>v^t_{k,l}$: The price $\underline{\hat{p}}^t_{k,l}$ charged to $d_k$ is equal to $\hat{p}^t_{k,l}$ because the $d_{k,l}^t$ has been the first pair in $Q_l^t$ when bidding truthfully. Thus, we have $\underline{\hat{u}}^t_{k,l}=\hat{u}^t_{k,l}$.
		\item (1.b) $b^t_{k,l}<v^t_{k,l}$: The price $\underline{\hat{p}}^t_{k,l}$ charged to $d_k$ is equal to $\hat{p}^t_{k,l}$ if the $d_{k,l}^t$ is still the first pair in $Q_l^t$ when bidding untruthfully ($b^t_{k,l}\geq\hat{p}^t_{k,l}$). Thus, we have $\underline{\hat{u}}^t_{k,l}=\hat{u}^t_{k,l}$. If $b^t_{k,l}<\hat{p}^t_{k,l}$, we have $s_l\notin\underline{H}^t_k$ and $\underline{\hat{u}}^t_{k,l}=0$. Thus, we have $\underline{\hat{u}}^t_{k,l}=0<\hat{u}^t_{k,l}$.
	\end{itemize}
	Therefore, the utility cannot be improved by bidding untruthfully to seller in $H^t_k$. For each seller $s_l\in S\backslash H^t_k$, we consider the following two sub-cases:
	\begin{itemize}
		\item (1.c) $d^t_{k,l}\notin D_c^t$ but $d^t_{k,l}\in D^t_*$: It has nothing to do with what the bid $b^t_{k,l}$ is if $a^t_l\geq a_{j_\phi}^t$. Thus, we have $\underline{\hat{u}}^t_{k,l}=\hat{u}^t_{k,l}=0$. Consider $a^t_l<a_{j_\phi}^t$ and $v^t_{k,l}<a^t_l$, the $d_k$ increases its bid $b^t_{k,l}$ such that $b^t_{k,l}\geq a_{j_\phi}^t$. Now, we get $\underline{\hat{u}}^t_{k,l}=(v^t_{k,l}-\underline{\hat{p}}^t_{k,l})\cdot r_k^t\leq(v^t_{k,l}-a_{j_\phi}^t)\cdot r_k^t<0$ if $s_l\in\underline{H}^t_k$, otherwise $\underline{\hat{u}}^t_{k,l}=0$. Thus, we have $\underline{\hat{u}}^t_{k,l}\leq\hat{u}^t_{k,l}=0$.
		\item (1.d) $d^t_{k,l}\in D_c^t$ but $s_l\notin H^t_k$: There is a pair $b^t_{k_1,l}=Q^t_l[1]$ such that $b^t_{k_1,l}\cdot r_{k_1}^t\geq v^t_{k,l}\cdot r_{k}^t$. To make $s_l$ be in $H^t_k$, the $d_k$ increases its bid such that $b^t_{k,l}\geq b^t_{k_1,l}$. Like this, the charged price will be $\underline{\hat{p}}^t_{k,l}=b^t_{k_1,l}\cdot (r_{k_1}^t/r_k^t)\geq v^t_{k,l}$. Now, we get $\underline{\hat{u}}^t_{k,l}\leq 0$ if $s_l\in\underline{H}^t_k$, otherwise $\underline{\hat{u}}^t_{k,l}=0$. Thus, we have $\underline{\hat{u}}^t_{k,l}\leq\hat{u}^t_{k,l}=0$.
	\end{itemize}

	\noindent
	\textbf{(2)} The $d_k\notin D_w^t$ if it bids truthfully: There is no such a seller $s_l$ that can satisfy $d^t_{k,l}=Q^t_l[1]$. For each seller $s_l\in S$, it can be analyzed like (1.c) and (1.d) before. Thus, we have $\underline{\hat{u}}^t_{k,l}\leq\hat{u}^t_{k,l}=0$.
	
	When $d_k\in D_w^t$ if it bids truthfully, its utility $\hat{u}^t_k=\max_{s_l\in H^t_k}\{\hat{u}^t_{k,l}\}\geq \max_{s_l\in \underline{H}^t_k}\{\underline{\hat{u}}^t_{k,l}\}=\underline{\hat{u}}^t_k$; When $d_k\notin D_w^t$ if it bids truthfully, $\hat{u}^t_k=0\geq\underline{\hat{u}}^t_k$. Therefore, the buyers are truthful. Next, For each buyer $s_l\in S$, we need to judge whether  $\bar{u}^t_l((c_l^t,q_l^t),\Omega^t_{-l})\geq\bar{u}^t_l((a_l^t,q_l^t),\Omega^t_{-l})$.
	
	\noindent
	\textbf{(3)} The $s_l\in S_w^t$ if it asks truthfully: Based on Algorithm \ref{a2}, we have $v^t_l<a_{j_\phi}^t$. For convenience, any notation $z^t_l$ and $\underline{z}^t_l$ refer to the concepts given by truthful ask $c^t_l$ and untruthful ask $a^t_l$. We consider the following two sub-cases:
	\begin{itemize}
		\item (3.a) $a^t_l\geq a^t_{j_\phi}$: It is easy to infer that $a^t_l\geq\underline{a}^t_{j_\phi}\geq a^t_{j_\phi}$. Thus, we have $\bar{u}^t_l>\underline{\bar{u}}^t_l=0$.
		\item (3.b) $a^t_l< a^t_{j_\phi}$: Now, we have $\underline{a}^t_{j_\phi}= a^t_{j_\phi}$ and $\sigma^t_{-1}(s_l)=\underline{\sigma}^t_{-1}(s_l)$ because all steps remain unchanged. Thus, we have $\bar{u}^t_l=\underline{\bar{u}}^t_l$.
	\end{itemize}

	\noindent
	\textbf{(4)} The $s_l\notin S_w^t$ if it asks truthfully: If losing because $c^t_l\geq a^t_{j_\phi}$, we consider the following two sub-cases:
	\begin{itemize}
		\item (4.a) $a^t_l\geq a^t_{j_\phi}$: Now, we have $\underline{a}^t_{j_\phi}= a^t_{j_\phi}$ and $a^t_l\geq \underline{a}^t_{j_\phi}$. Thus, we have $\bar{u}^t_l=\underline{\bar{u}}^t_l=0$.
		\item (4.b) $a^t_l< a^t_{j_\phi}$: It is easy to infer that $a^t_l\leq\underline{a}^t_{j_\phi}\leq a^t_{j_\phi}$. If there is no buyer assigned to it, we have $\bar{u}^t_l=\underline{\bar{u}}^t_l=0$; Else, the rewarded payment $\underline{\bar{p}}^t_l=\underline{a}^t_{j_\phi}\leq c^t_l$. Thus, we have $\bar{u}^t_l=0\geq\underline{\bar{u}}^t_l$.
	\end{itemize}
	
	\noindent
	If $c^t_l<a^t_{j_\phi}$ but still losing, it is either $s_l\notin S^t_c$ or $\sigma^t(d_{k_1})\neq s_l$ where $d^t_{k_1,l}=Q^t_l[1]$. We consider the following two sub-cases:
	\begin{itemize}
		\item (4.c) $a^t_l\geq a^t_{j_\phi}$: It is easy to infer that $a^t_l\geq\underline{a}^t_{j_\phi}\geq a^t_{j_\phi}$. Thus, we have $\bar{u}^t_l=\underline{\bar{u}}^t_l=0$.
		\item (4.d) $a^t_l< a^t_{j_\phi}$: Now, we have $\underline{a}^t_{j_\phi}= a^t_{j_\phi}$ and $a^t_l< \underline{a}^t_{j_\phi}$. It is either $s_l\notin \underline{S}^t_c$ or $\underline{\sigma}^t(d_{k_1})\neq s_l$ where $d^t_{k_1,l}=\underline{Q}^t_l[1]$ as well. Thus, we have $\bar{u}^t_l=\underline{\bar{u}}^t_l=0$.
	\end{itemize}

	Therefore, both buyers and sellers are truthful, which leads to the MIDA is truthful.
\end{proof}

\begin{thm}
	The MIDA is individually rational, budget balanced, computationally efficient, and truthful.
\end{thm}
\begin{proof}
	It can be derived from Lemma 1 to Lemma 4.
\end{proof}

\subsection{Walk-through Example and Inference Attack}
	We give a walk-through example to demonstrate how our MIDA mechanism works. The bid information of buyers and ask information of sellers are shown in Table \ref{table1}. In the time slot $t$, we assume each buyer satisfies the Constraint (5d) and each player bids (asks) truthfully. In the MIDA-WCD process, the threshold is $a^t_{j_\phi}=a^t_3=4$, thereby we have $S^t_c=\{s_2,s_5,s_6\}$ and $D^t_c=\{d^t_{1,2},d^t_{2,5},d^t_{3,6},d^t_{4,2},d^t_{4,5},d^t_{5,6}\}$. In the MIDA-AP process, for each winning seller candidate, we have $Q_2^t=\langle d^t_{4,2},d^t_{1,2}\rangle$ since $b^t_{4,2}\cdot r^t_4=24>b^t_{1,2}\cdot r^t_1=20$. Similarly, we have $Q_5^t=\langle d^t_{4,5},d^t_{2,5}\rangle$ and $Q_6^t=\langle d^t_{3,6},d^t_{5,6}\rangle$. The prices charged to them should be $\hat{p}^t_{4,2}=\max\{a^t_{j_\phi},b^t_{1,2}\cdot(r^t_1/r^t_4)\}=5$, $\hat{p}^t_{4,5}=\max\{a^t_{j_\phi},b^t_{2,5}\cdot(r^t_2/r^t_4)\}=4$, and $\hat{p}^t_{3,6}=\max\{a^t_{j_\phi},b^t_{5,6}\cdot(r^t_5/r^t_3)\}=4$. For the buyer $d_4$, we have $\hat{u}^t_{4,2}=(v^t_{4,2}-\hat{p}^t_{4,2})\cdot r^t_4=4$ and $\hat{u}^t_{4,5}=(v^t_{4,5}-\hat{p}^t_{4,5})\cdot r^t_4=8$. To maximize its utility, the buyer $d_4$ is assigned to seller $s_5$. The published auction result is $D^t_w=\{d_3,d_4\}$, $S^t_w=\{s_5,s_6\}$, $\sigma^t(d_3)=s_6$, $\sigma^t(d_4)=s_5$, $\hat{P}^t_w=\{\hat{p}^t_3,\hat{p}^t_4\}=\{4,4\}$, $\bar{P}^t_w=\{\bar{p}^t_2,\bar{p}^t_6\}=\{4,4\}$, and $\{r_3^t,r^t_4\}=\{6,4\}$. The social welfare can be written as 
	\begin{equation}
		\sum_{d_k\in D^t_w}\left(b^t_{k,\sigma^t(d_k)}-a^t_{\sigma^t(d^k)}\right)\cdot r^t_k=24
	\end{equation}

	\begin{table}[!t]
		\renewcommand{\arraystretch}{1.2}
		\caption{A walk-through example with 5 buyers and 7 sellers.}
		\label{table1}
		\centering
		\begin{tabular}{c|ccccccc|c}
			\hline
			$b_{k,l}^t(v_{k,l}^t)$ & $s_1$ & $s_2$ & $s_3$ & $s_4$ & $s_5$ & $s_6$ & $s_7$ & $r_k^t$\\
			\hline
			$d_1$ & 0 & 4 & 0 & 5 & 2 & 0 & 0 & 5\\
			$d_2$ & 2 & 0 & 0 & 0 & 5 & 1 & 0 & 2\\
			$d_3$ & 7 & 0 & 5 & 0 & 0 & 4 & 0 & 6\\
			$d_4$ & 0 & 6 & 4 & 0 & 6 & 0 & 0 & 4\\
			$d_5$ & 0 & 0 & 0 & 2 & 0 & 4 & 5 & 3\\
			\hline
			$a^t_l(c^t_l)$ & 6 & 1 & 4 & 5 & 3 & 2 & 5 & -\\
			\hline
			$g^t_l$ & 3 & 7 & 6 & 5 & 8 & 7 & 6 & -\\
			\hline
		\end{tabular}
	\end{table}

	\textbf{Case 1:} As mentioned before, there are two kinds of inference attack. First, we consider a seller adversary infers the unit asks of other sellers in the time slot $t+1$. We assume that seller $s_3$ is the adversary who gives an untruthful ask $a^{t+1}_3=2$ and other players remain the same as the last slot. Now, we have $\underline{a}^{t+1}_{j_\phi}=a^{t+1}_5=3$, $\underline{S}^{t+1}_c=\{s_2,s_3,s_6\}$, and $\underline{D}^{t+1}_c=\{d^{t+1}_{1,2},d^{t+1}_{3,3},d^{t+1}_{3,6},d^{t+1}_{4,2},d^{t+1}_{4,3},d^{t+1}_{5,6}\}$. The published auction result will be $\underline{D}^{t+1}_w=\{d_3,d_4\}$, $\underline{S}^{t+1}_w=\{s_2,s_3\}$, $\underline{\sigma}^{t+1}(d_3)=s_3$, $\underline{\sigma}^{t+1}(d_4)=s_2$, $\underline{\hat{P}}^{t+1}_w=\{\hat{p}^{t+1}_3,\hat{p}^{t+1}_4\}=\{4,5\}$, $\underline{\bar{P}}^{t+1}_w=\{\bar{p}^{t+1}_2,\bar{p}^{t+1}_6\}=\{4,4\}$, and $\{r_3^{t+1},r^{t+1}_4\}=\{6,4\}$. Now, we can observe that the seller $s_5$ is not a winning seller. Besides, the seller $s_5$ is not in current seller candidate set $\underline{S}^{t+1}_c$ because the buyer $d_4$ will be assigned to it if $s_5\in \underline{S}^{t+1}_c$ based on the result of last slot. Since seller $s_3$ decreases its ask, seller $s_5$ is removed from the seller candidate set. According to current threshold payment to sellers, it is easy to infer that $a^{t+1}_5=\underline{a}^{t+1}_{j_\phi}=3$. Then, the privacy of seller $s_5$ has been threatened.
	
	\textbf{Case 2:} Next, we consider a buyer adversary infers the unit bids of other buyers in the time slot $t+1$. We assume that buyer $d_1$ is the adversary who gives an untruthful bid $b^{t+1}_{1,5}=6$ and other players remain the same as the last slot. Now, we have $\underline{a}^{t+1}_{j_\phi}=a^{t+1}_3=4$, $\underline{S}^{t+1}_c=\{s_2,s_5,s_6\}$, and $\underline{D}^{t+1}_c=\{d^{t+1}_{1,2},d^{t+1}_{1,5},d^{t+1}_{2,5},d^{t+1}_{3,6},d^{t+1}_{4,2},d^{t+1}_{4,5},d^{t+1}_{5,6}\}$. The published auction result will be $\underline{D}^{t+1}_w=\{d_1,d_3,d_4\}$, $\underline{S}^{t+1}_w=\{s_2,s_5,s_6\}$, $\underline{\sigma}^{t+1}(d_1)=s_5$, $\underline{\sigma}^{t+1}(d_3)=s_6$, $\underline{\sigma}^{t+1}(d_4)=s_2$, $\underline{\hat{P}}^{t+1}_w=\{\hat{p}^{t+1}_1,\hat{p}^{t+1}_3,\hat{p}^{t+1}_4\}=\{4.8,4,5\}$, $\underline{\bar{P}}^{t+1}_w=\{\bar{p}^{t+1}_2,\bar{p}^{t+1}_5,\bar{p}^{t+1}_6\}=\{4,4,4\}$, and $\{r_1^{t+1},r_3^{t+1},r^{t+1}_4\}=\{5,6,4\}$. Now, we can observe that the buyer $d_4$ is assigned to seller $s_2$ instead of $s_5$. Thus, buyer $d^{t+1}_{1,5}$ replaces the top position of $d^{t+1}_{4,5}$ in the sorted list $Q_5^{t+1}$. According to the price charged to $d_1$, we have $\hat{p}^{t+1}_1=\hat{p}^{t+1}_{1,5}=b^{t+1}_{4,5}\cdot (r^{t+1}_4/r^{t+1}_1)$. It is easy to infer that $b^{t+1}_{4,5}=\hat{p}^{t+1}_1\cdot (r^{t+1}_1/r^{t+1}_4)=6$. Then, the privacy of buyer $d_4$ has been threatened.

\section{Differentially Private Online Double Auction Design}
To protect the privacy of both buyers and sellers, we improve our MIDA mechanism by using the technology of differential privacy. We first introduce several important concepts about differential privacy.

\subsection{Differential Privacy}
    Differential privacy \cite{dwork2008differential} is a technology to guarantee that an adversary is not capable of distinguishing between two neighboring inputs with high probability. The neighboring databases means two data sets $O=\{o_1,o_2,\cdots,o_{|O|}\}$ and $O'=\{o'_1,o'_2,\cdots,o'_{|O|}\}$ which have exactly one different element. In differentially private protection, it is possible that two neighboring inputs have the same or similar output. Thus, adversaries are hard to infer other private inputs according to public query results. Let us look at its definition.
    \begin{defn}[Differential Privacy]
    	An algorithm (query function) $f$ gives $\varepsilon$-differential privacy if and only if, for any two neighboring inputs $O$ and $O'$, we have
    	\begin{equation}
    		\Pr[f(O)\in R]\leq\exp(\varepsilon)\cdot\Pr[f(O')\in R]
    	\end{equation}
    	where $R$ is a fixed range such that $R\subseteq Range(f)$ and $\varepsilon$ is called privacy budget.
    \end{defn}
	\noindent
	The privacy budget is a parameter used to control the degree of privacy protection that an algorithm gives. Usually, a smaller privacy budget implies a stronger privacy protection. The sensitivity of an algorithm $f$ quantifies the magnitude of noise that is needed to protect the data from adversaries.
	\begin{defn}[Sensitivity]
		The $\ell_1$-sensitivity of an algorithm $f$ is defined as
		\begin{equation}
			\Delta f=\max_{O,O'\in dom(f)}||f(O)-f(O')||_1
		\end{equation}
	\end{defn}
	Based on this definition, the sensitivity is an upper bound we need to perturb the output of $f$ to protect privacy. The noise is generally sampled from a Laplace distribution. A random variable $X$ subjected to Laplace distribution, denoted by $X\sim\text{Lap}(\mu,b)$, has a probability density function
	\begin{equation}
		\text{Lap}(x|\mu,b)=\frac{1}{2b}\exp\left(-\frac{|x-\mu|}{b}\right)
	\end{equation}
	where $\mu$ is the center point and $b$ is the scaling factor \cite{zhu2017preliminary}. Now, we can define the Laplace mechanism. The Laplace mechanism runs an algorithm $f$ directly and then adds a Laplace noise sampled from the Laplace distribution.
	\begin{defn}[Laplace Mechanism]
		Given an algorithm (query function) $f:dom(f)\rightarrow\mathbb{R}$, the Laplace mechanism $\mathcal{M}_L(x,f,\varepsilon)$ can be defined as
		\begin{equation}
			\mathcal{M}_L(x,f,\varepsilon)=f(x)+\text{\rm Lap}\left(0,\frac{\Delta f}{\varepsilon}\right)
		\end{equation}
		where $x\in dom(f)$ and ${\rm Lap}(\Delta f/\varepsilon)$ is a random noise sampled from the Laplace distribution.
	\end{defn}

\begin{algorithm}[!t]
	\caption{\text{MIDA-DP $(\Omega^t)$}}\label{a4}
	\begin{algorithmic}[1]
		\renewcommand{\algorithmicrequire}{\textbf{Input:}}
		\renewcommand{\algorithmicensure}{\textbf{Output:}}
		\REQUIRE $\Omega^t=(\{\mathcal{B}_i^t\}_{d_i\in D},\{\mathcal{A}_j^t\}_{s_j\in S})$, $\varepsilon$
		\ENSURE $D_w^t,S_w^t,\sigma^t,\hat{P}_w^t,\bar{P}_w^t$
		\STATE $(D_c^t,S_c^t,\tilde{a}_{j_\phi}^t)\leftarrow$ MIDA-WCD-DP $(\Omega^t,\varepsilon)$
		\STATE $(D_w^t,S_w^t,\sigma^t,\hat{P}_w^t,\bar{P}_w^t)\leftarrow$ MIDA-AP $(\Omega^t,D_c^t,S_c^t,\tilde{a}_{j_\phi}^t)$
		\RETURN $(D_w^t,S_w^t,\sigma^t,\hat{P}_w^t,\bar{P}_w^t)$ 
	\end{algorithmic}
\end{algorithm}

\subsection{Algorithm Design}
Here, we make some minor changes to the MIDA mechansim to ensure the private security of both buyers and sellers. The differentially private multi-item double auction (MIDA-DP) mechanism is shown in Algorithm \ref{a4}, whose winning candidate determination part (MIDA-WCD-AP) is shown in Algorithm \ref{a5}. Shown as Algorithm \ref{a5}, we select the meadian $a^t_{j_\phi}$ as a threshold first and add an Laplace noise sampled from the Laplace distribution ${\rm Lap}(0,\Delta_1/\varepsilon)$ to get an updated threshold $\tilde{a}_{j_\phi}^t$. Then, we use the updated threshold to select candidate winning buyer and seller set.

\subsection{Theoretical Analysis and Privacy Protection}
First, we give the analysis about privacy protection in our MIDA-DP mechanism.

\begin{lem}
	The MIDA-DP gives $\varepsilon$-differential privacy for the asks of sellers.
\end{lem}
\begin{proof}
	The process of MIDA-WCD-DP shown as Algorithm \ref{a5} can be regarded as an algorithm $f:[a_{min},a_{max}]^m\rightarrow[a_{min},a_{max}]$ where the input is the asks of sellers $\vec{a}^t=\{a_1^t,a_2^t,\cdots,a^t_m\}\in[a_{min},a_{max}]^m$ and the output is the threshold $a^t_{j_\phi}\in[a_{min},a_{max}]$ in the time slot $t$. Let us consider two neighboring inputs $\vec{a}^t$ and $\vec{a}^t_*$ that has exactly one different ask. For any value $\tilde{a}_{j_\phi}^t\in\mathbb{R}$, we have
	\begin{align}
		&\frac{\Pr[\mathcal{M}_L(\vec{a}^t,f,\varepsilon)=\tilde{a}_{j_\phi}^t]}{\Pr[\mathcal{M}_L(\vec{a}^t_*,f,\varepsilon)=\tilde{a}_{j_\phi}^t]}\nonumber\\
		&\quad=\frac{\exp(-{\varepsilon|f(\vec{a}^t)-\tilde{a}_{j_\phi}^t|}/{\Delta_1})}{\exp(-{\varepsilon|f(\vec{a}^t_*)-\tilde{a}_{j_\phi}^t|}/{\Delta_1})}\nonumber\\
		&\quad=\exp\left({\varepsilon(|f(\vec{a}^t)-\tilde{a}_{j_\phi}^t|-|f(\vec{a}^t_*)-\tilde{a}_{j_\phi}^t|)}/{\Delta_1}\right)\nonumber\\
		&\quad\leq \exp\left(\varepsilon|f(\vec{a}^t)-f(\vec{a}^t_*)|/{\Delta_1}\right)\nonumber\\
		&\quad\leq \exp\left(\varepsilon{\Delta_1}/{\Delta_1}\right)\nonumber\\
		&\quad=\exp(\varepsilon)\nonumber
	\end{align}
	becuase we have known $|f(\vec{a}^t)-f(\vec{a}^t_*)|\leq\Delta_1$. By symmetry, we have ${\Pr[\mathcal{M}_L(\vec{a}^t,f,\varepsilon)=\tilde{a}_{j_\phi}^t]}/{\Pr[\mathcal{M}_L(\vec{a}^t_*,f,\varepsilon)=\tilde{a}_{j_\phi}^t]}\geq\exp(-\varepsilon)$ easily.
\end{proof}

Then, we give the qualitative analysis about privacy protection based on the aforementioned two cases of inference attack discussed in Sec. IV-D. For the Case 1, even though it can infer that the seller $s_5$ is not in current winning seller candidate set, the adversary cannot conclude with high confidence that $a^{t+1}_5$ is the maximum one higher than or equal to the threshold. Besides, the threshold $\tilde{a}_{j_\phi}^t$ has been perturbed. Even though $a^{t+1}_5$ is the maximum one higher than or equal to the threshold, the adversary cannot know what the $a^{t+1}_5$ is. For the Case 2, because of the uncertainty of threshold $a^t_{j_\phi}$ and $\underline{a}^{t+1}_{j_\phi}$, the adversary cannot make sure that $d^{t+1}_{4,5}$ is in current winning buyer candidate set. Thus, it is hard to infer that $\hat{p}^{t+1}_{1,5}=b^{t+1}_{4,5}\cdot (r^{t+1}_4/r^{t+1}_1)$. The adversary cannot know what the $b^{t+1}_{4,5}$ is. The privacy security of both buyers and sellers has been guaranteed.

\begin{algorithm}[!t]
	\caption{\text{MIDA-WCD-DP $(\Omega^t)$}}\label{a5}
	\begin{algorithmic}[1]
		\renewcommand{\algorithmicrequire}{\textbf{Input:}}
		\renewcommand{\algorithmicensure}{\textbf{Output:}}
		\REQUIRE $\Omega^t=(\{\mathcal{B}_i^t\}_{d_i\in D},\{\mathcal{A}_j^t\}_{s_j\in S})$, $\varepsilon$
		\ENSURE $D_c^t,S_c^t,\tilde{a}_{j_\phi}^t$
		\STATE $D_c^t\leftarrow\emptyset,S_c^t\leftarrow\emptyset$
		\STATE Construct a set $D^t_*=\{d^t_{k,l}:d^t_{k,l}\text{ satisfies (7) (8)}\}$
		\STATE Sort the sellers such that $S^t_*=\langle s_{j_1}^t,s_{j_2}^t,\cdots,s_{j_m}^t\rangle$ where $a_{j_1}^t\leq a_{j_2}^t\leq\cdots\leq a_{j_m}^t$	
		\STATE Find the median ask $a_{j_\phi}^t$ of $S^t_*$, $\phi=\left\lceil\frac{m+1}{2}\right\rceil$
		\STATE $\Delta_1\leftarrow(a_{max}-a_{min})$
		\STATE $\tilde{a}_{j_\phi}^t\leftarrow a_{j_\phi}^t+\text{Lap}(0,\Delta_1/\varepsilon)$
		\FOR {each $d^t_{k,l}\in D^t_*$}
		\IF {$b^t_{k,l}\geq \tilde{a}_{j_\phi}^t$ and $a_l^t<\tilde{a}_{j_\phi}^t$}
		\STATE $D_c^t\leftarrow D_c^t\cup\{d^t_{k,l}\}$
		\IF {$s_l\notin S_c^t$}
		\STATE $S_c^t\leftarrow S_c^t\cup\{s_l\}$
		\ENDIF
		\ENDIF
		\ENDFOR
		\RETURN $(D_c^t,S_c^t,\tilde{a}_{j_\phi}^t)$
	\end{algorithmic}
\end{algorithm}

\begin{thm}
	The MIDA-DP is individually rational, budget balanced, computationally efficient, and truthful. Moreover, it protects the privacy of buyers and sellers.
\end{thm}
\begin{proof}
	The MIDA-PD satisfies the design rationales from Definition 1 to Definition 4 by similar proofs from Lemma 1 to Lemma 4. According to Lemma 5, it gives $\varepsilon$-differential privacy to the sellers. Because the winning buyers and their charged prices are dependent on the winning seller candidate set, the privacy of buyers would be protected as well.
\end{proof}

\section{Problem Extension}
For the previous problem defined in Sec. \uppercase\expandafter{\romannumeral3}-B, it is just a special case where each edge server (seller) can serve at more one device in a slot. Actually, an edge server could serve more than one devices. Thus, in a more general scenario, the Constraint (5a) should be removed, and thus, the function $\sigma^t$ from $D_w^t$ to $S_w^t$ is not a bijection, but a many-to-one mapping. To distinguish from MIDA mechanism, the algorithm to solve this general case is named as MIDA-G mechanism shown in Algorithm \ref{a6}. Its winning candidate determination is the same as MIDA-WCD shown in Algorithm \ref{a2}, but its assignment \& price (MIDA-G-AP) is shown in Algorithm \ref{a7}. Here, it exists a tentative set $Q_j^t$ for each $s_j\in S_c^t$ where each buyer in $Q_j^t$ could be assigned to seller $s_j$. Naturally, we have $\sum_{d_i\in Q_j^t}r_i^t \leq q^t_j$. The charged price given in line 9 and 17 is used to guarantee the truthfulness. 

\begin{algorithm}[!t]
	\caption{\text{MIDA-G $(\Omega^t)$}}\label{a6}
	\begin{algorithmic}[1]
		\renewcommand{\algorithmicrequire}{\textbf{Input:}}
		\renewcommand{\algorithmicensure}{\textbf{Output:}}
		\REQUIRE $\Omega^t=(\{\mathcal{B}_i^t\}_{d_i\in D},\{\mathcal{A}_j^t\}_{s_j\in S})$
		\ENSURE $D_w^t,S_w^t,\sigma^t,\hat{P}_w^t,\bar{P}_w^t$
		\STATE $(D_c^t,S_c^t,a_{j_\phi}^t)\leftarrow$ MIDA-WCD $(\Omega^t)$
		\STATE $(D_w^t,S_w^t,\sigma^t,\hat{P}_w^t,\bar{P}_w^t)\leftarrow$ MIDA-G-AP $(\Omega^t,D_c^t,S_c^t,a_{j_\phi}^t)$
		\RETURN $(D_w^t,S_w^t,\sigma^t,\hat{P}_w^t,\bar{P}_w^t)$ 
	\end{algorithmic}
\end{algorithm}

By similar induction as the MIDA mechanism, we also have that the MIDA-G mechanism is individually rational, budget balanced, computationally efficient, and truthful. The differentially private strategy shown in Sec. \uppercase\expandafter{\romannumeral5} can be used in our MIDA-G mechanism to protect the privacy of buyers and sellers for the same reason. Thus, the MIDG-G-DP mechanism can be formulated by replacing MIDA-WCD in Algorithm \ref{a6} with MIDA-WCD-DP shown as Algorithm \ref{a5}.

\begin{algorithm}[!t]
	\caption{\text{MIDA-G-AP $(\Omega^t,D_c^t,S_c^t,a_{j_\phi}^t)$}}\label{a7}
	\begin{algorithmic}[1]
		\renewcommand{\algorithmicrequire}{\textbf{Input:}}
		\renewcommand{\algorithmicensure}{\textbf{Output:}}
		\REQUIRE $\Omega^t,D_c^t,S_c^t,a_{j_\phi}^t$
		\ENSURE $D_w^t,S_w^t,\sigma^t,\hat{P}_w^t,\bar{P}_w^t$
		\STATE $D_w^t\leftarrow\emptyset,S_w^t\leftarrow\emptyset,\hat{P}_w^t\leftarrow\emptyset,\bar{P}_w^t\leftarrow\emptyset$
		\STATE Create a sorted list $Q_j^t=\langle d^t_{k,j}:d^t_{k,j}\in D_c^t\rangle$ for each $s_j\in S_c^t$ such that $b_{k_1,j}^t\cdot r_{k_1}\geq b_{k_2,j}^t\cdot r_{k_2}\geq\cdots$
		\FOR {each $s_l\in S_c^t$}
		\IF {$\sum_{d^t_{k,l}\in Q_l^t}r_k^t\leq q_l^t$}
		\FOR {each $d^t_{k,l}\in Q_l^t$}
		\IF {$d_k\in D_w^t$}
		\STATE $D^t_w\leftarrow D^t_w\cup\{d_k\}$
		\ENDIF
		\STATE $\hat{p}^t_{k,l}\leftarrow a^t_{j_\phi}$
		\ENDFOR
		\ELSE
		\STATE Let $k_m$ be the maximum value such that it satisfies contraint $\sum_{d^t_{k,l}\in Q_l^t[1,\cdots,k_m]}r_k\leq q_l^t$
		\FOR {each $d^t_{k,l}\in Q_l^t[1,\cdots,k_m]$}
		\IF {$d_k\in D_w^t$}
		\STATE $D^t_w\leftarrow D^t_w\cup\{d_k\}$
		\ENDIF
		\STATE $\hat{p}^t_{k,l}\leftarrow\max\{a^t_{j_\phi},b^t_{k_m+1,l}\cdot(r^t_{k_m+1}/r^t_k)\}$
		\ENDFOR
		\STATE $Q_l^t\leftarrow Q_l^t[1,\cdots,k_m]$
		\ENDIF
		\ENDFOR
		\FOR {each $d_k\in D_w^t$}
		\STATE $H^t_k=\{s_l:s_l\in S_c^t,d^t_{k,l}\in Q_l^t\}$
		\STATE Find $s_{l'}\leftarrow\arg\max_{s_l\in H_k^t}\{(b_{k,l}^t-\hat{p}_{k,l}^t)\cdot r_k^t\}$
		\STATE $\sigma^t(d_k)\leftarrow s_{l'}$
		\STATE $\hat{p}^t_k\leftarrow\hat{p}^t_{k,l'}$, $\hat{P}_w^t\leftarrow\hat{P}_w^t\cup\{\hat{p}^t_k\}$
		\IF {$s_{l'}\notin S_w^t$}
		\STATE $S_w^t\leftarrow S_w^t\cup\{s_{l'}\}$
		\STATE $\bar{p}^t_{l'}\leftarrow a_{j_\phi}^t$, $\bar{P}_w^t\leftarrow\bar{P}_w^t\cup\{\bar{p}^t_{l'}\}$
		\ENDIF
		\ENDFOR
		\RETURN $(D_w^t,S_w^t,\sigma^t,\hat{P}_w^t,\bar{P}_w^t)$
	\end{algorithmic}
\end{algorithm}

\section{Numerical Simulations}
In this section, we construct a virtual scenario to integrate a blockchain system with edge servers. The, we implement and evaluate our MIDA mechanism in detail.

\subsection{Simulation Setup}
We hypothesize an area with $\alpha\times \alpha$ $km^2$, where there are $n$ devices (blockchain nodes, buyers) and $m$ edge servers (sellers). We default by $n=m$ and they are distributed uniformly in this area. Given a device $d_i\in D$ and an edge server $s_j\in S$, their coordinates in this area are denoted by $(x_i,y_i)$ and $(x_j,y_j)$, which implies their positions. The distance between device $d_i$ and edge server $s_j$ can be defined as $Dist(d_i,s_j)$. We have
\begin{equation}
	Dist(d_i,s_j)=\sqrt{(x_i-x_j)^2+(y_i-y_j)^2}
\end{equation}
Here, we give a parameter $\gamma$ such that $\{d_i,s_j\}\in E$ if $Dist(d_i,s_j)\leq\gamma$, otherwise $\{d_i,s_j\}\notin E$. According to the definitions in Sec. \uppercase\expandafter{\romannumeral3}-B, in the time slot $t\in I$, we assume the unit cost $c_j^t$ of each seller $s_j\in S$ is distributed uniformly in interval $[0,1]$. Similarly, the unit valuation $v^t_{i,j}$ is also distributed uniformly in interval $[0,1]$ if $Dist(d_i,s_j)\leq\gamma$ and $v^t_{i,j}=0$ if $Dist(d_i,s_j)>\gamma$. The settings of computing power $r_i^t$ requested by buyer $d_i\in D$ and computing power $g_j^t$ provided by seller $s_j\in S$ will be introduced later.

\subsection{Simulation Results and Analysis}
\textbf{Part 1: }We consider a static MIDA mechanism in any time slot $t\in I$ based on the above settings, where we assume that $Dist(d_i,s_j)\leq\gamma$ and $r_i^t\leq g_j^t$ for any $d_i\in D$ and $s_j\in S$. Here, the $r_i^t$ is distributed uniformly in interval $[0,10]$. Moreover, the Constraint (8) has been satisfied. Here, we give $\alpha=10$ and $n=m=10$. Thus, we have $D=\{d_1,d_2,\cdots,d_{10}\}$ and $S=\{s_1,s_2,\cdots,s_{10}\}$. We use this simplified scene to evaluate whether it satisfies individual rationality, budget balance, and truthfulness.

\noindent
\textbf{(1.a) Individual Rationality and Budget Balance: }Fig. \ref{fig2} shows the auction results and their individual rationality, where the $d_9(s_7)$ means $\sigma^t(d_9)=s_7$. Shown as Fig. \ref{fig2}, we have $D_w^t=\{d_9,d_2,d_{10},d_3\}$ and $S_w^t=\{s_7,s_4,s_9,s_3\}$. The price charged to each winning buyer is less than its bid and the payment rewarded to each winning seller is more than its ask. Thus, individual rationality can be held. According to the definition of budget balance shown as Inequ. (\ref{eq9}), the budget balance can be held obviously.

\begin{figure}[!t]
	\centering
	\includegraphics[width=2.5in]{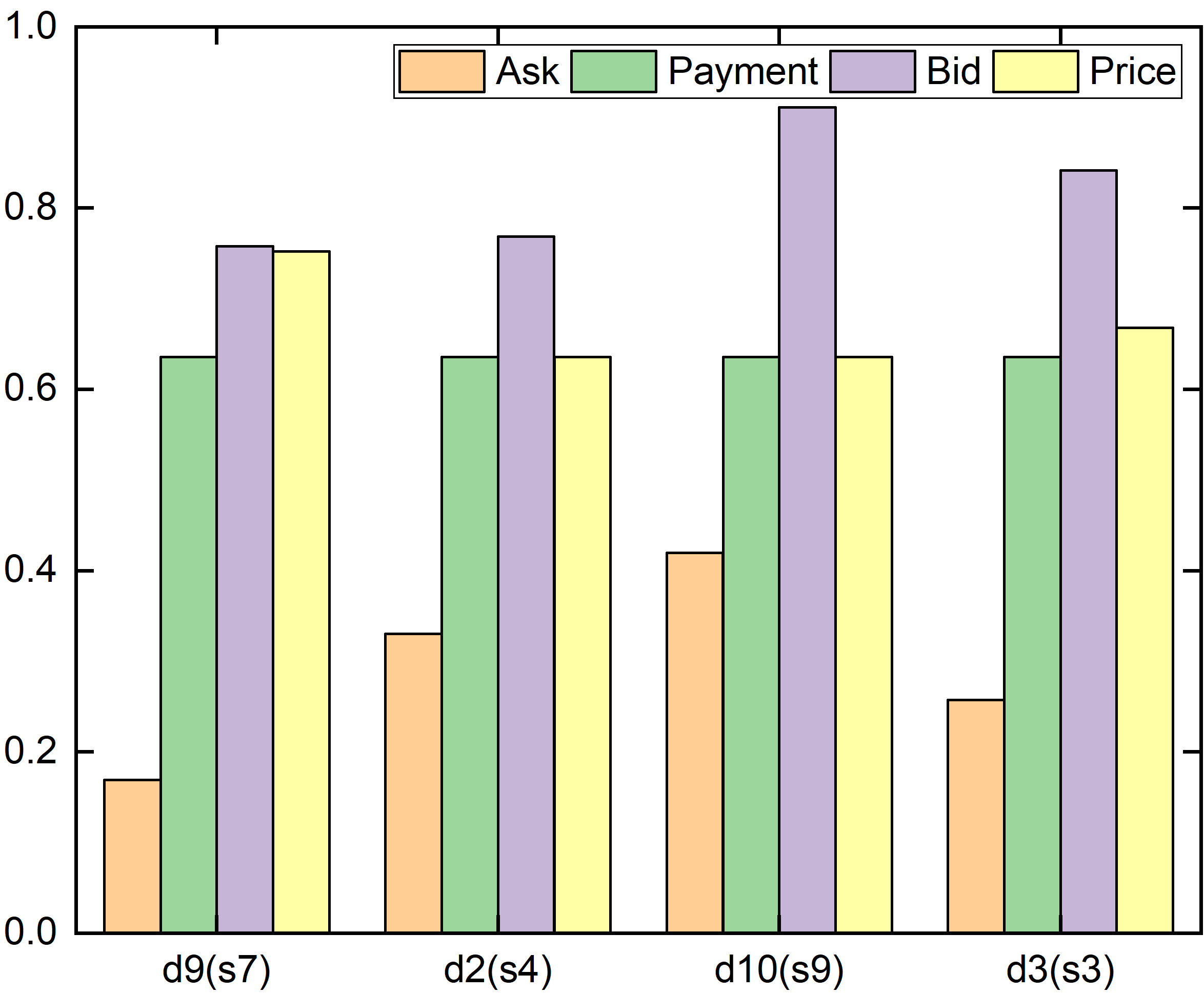}
	\centering
	\caption{The auction results obtained by MIDA and their individual rationality.}
	\label{fig2}
\end{figure}

\begin{figure}[!t]
	\centering
	\subfigure[Buyer $d_2\in D_w^t$]{
		\includegraphics[width=0.48\linewidth]{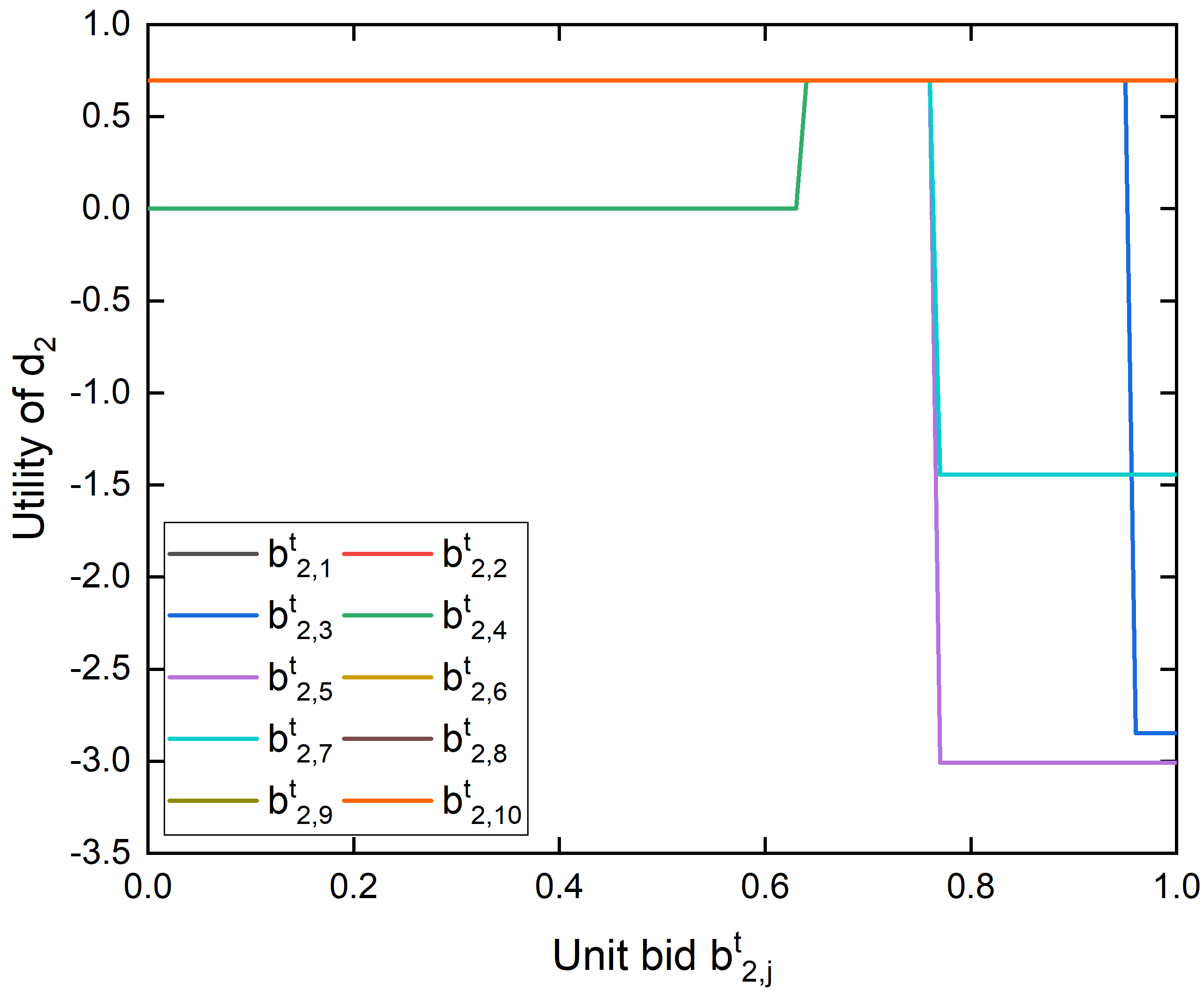}
	}%
	\subfigure[Buyer $d_7\notin D_w^t$]{
		\includegraphics[width=0.48\linewidth]{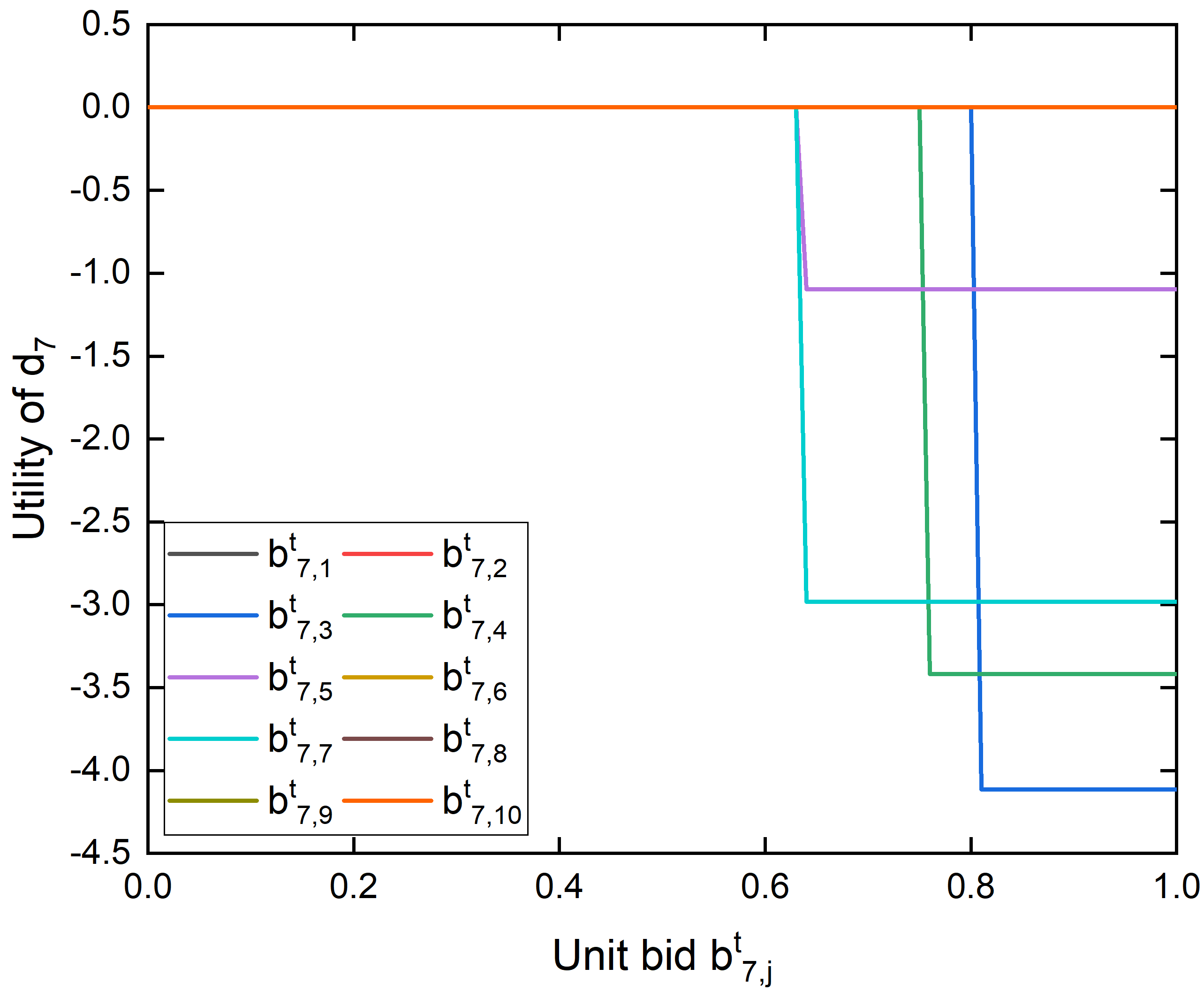}
	}%
	
	\subfigure[Seller $s_7\in S_w^t$]{
		\includegraphics[width=0.48\linewidth]{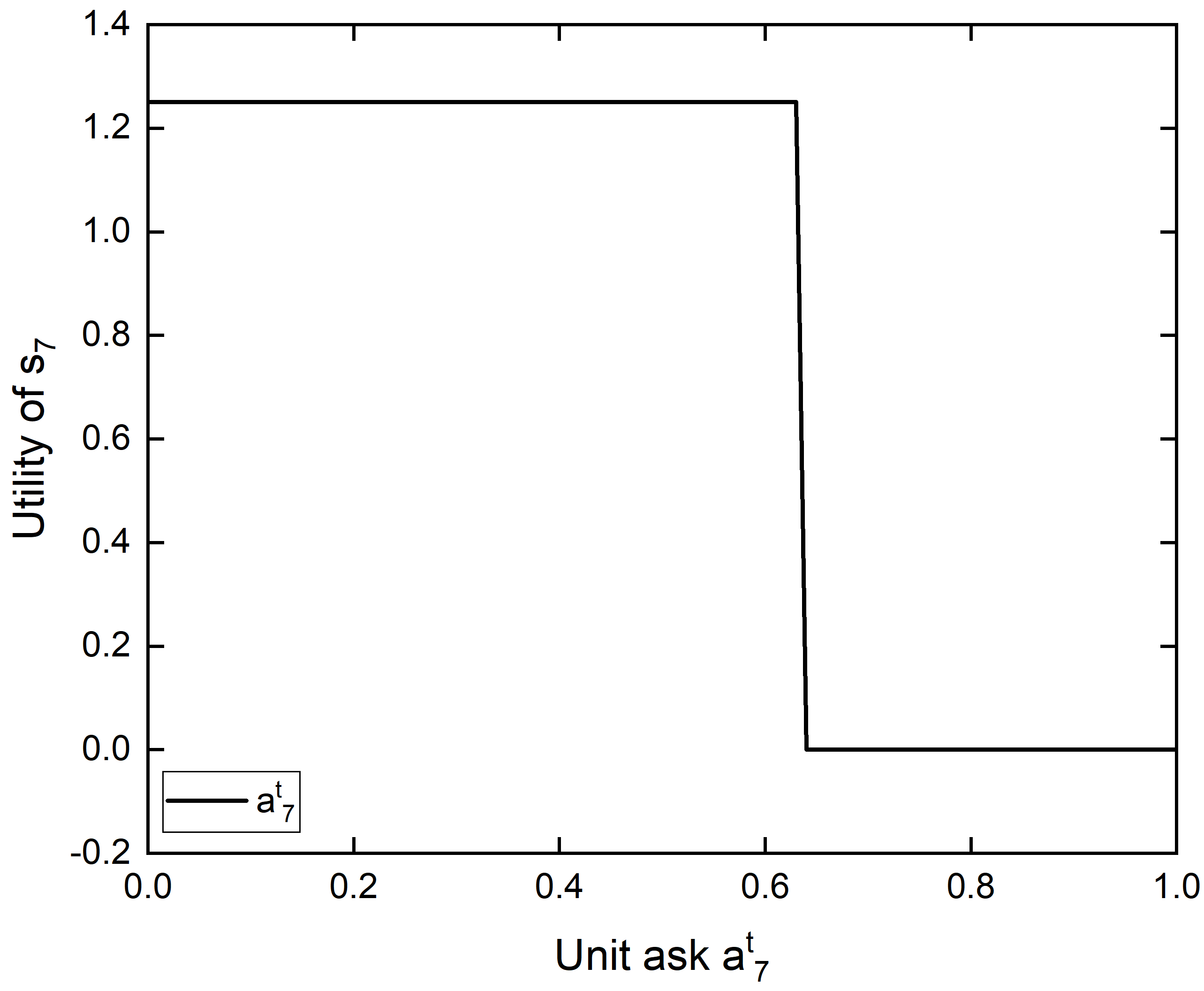}
	}%
	\subfigure[Seller $s_1\notin S_w^t$]{
		\includegraphics[width=0.48\linewidth]{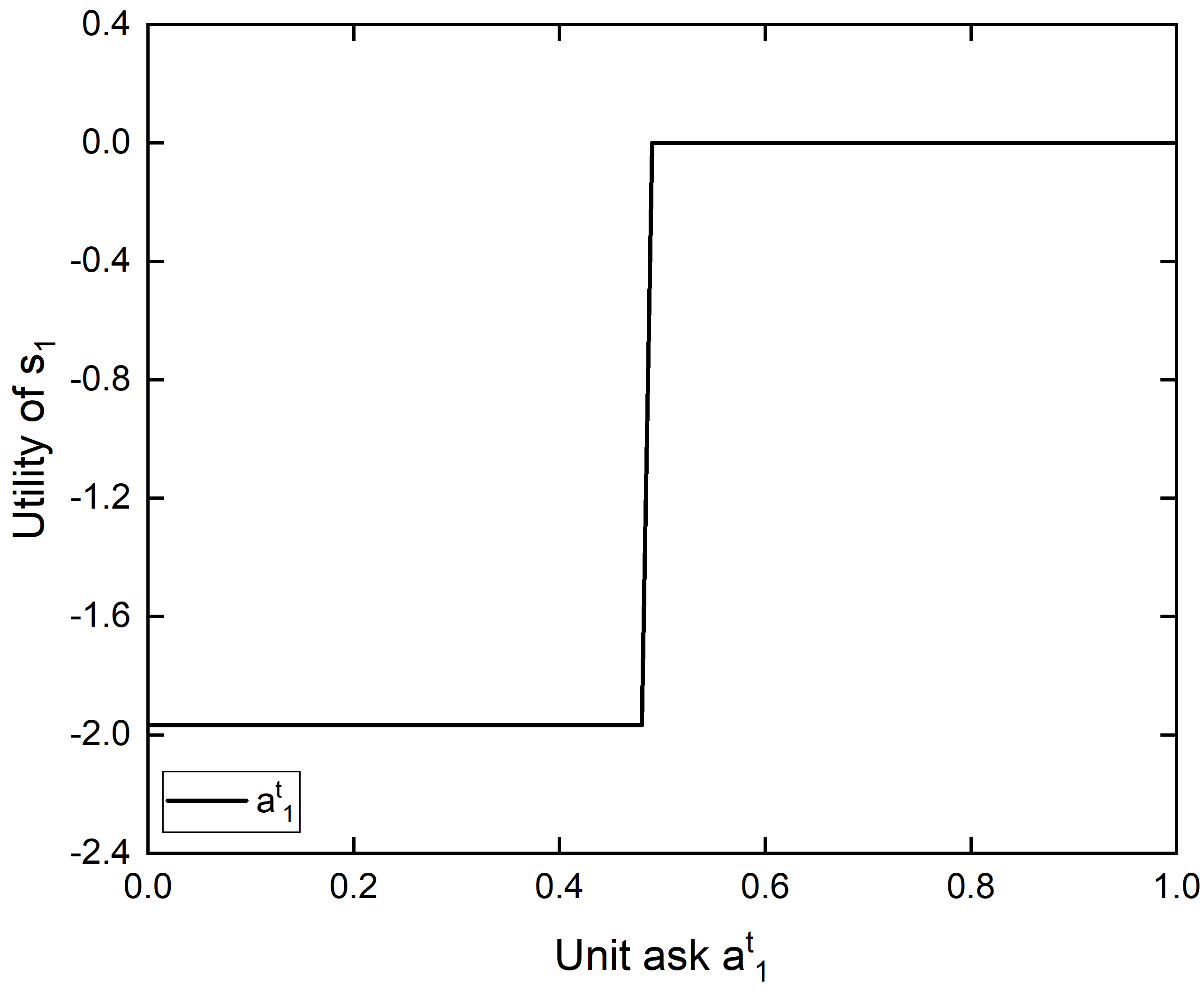}
	}%
	\centering
	\caption{The truthfulness of buyers and sellers in MIDA.}
	\label{fig3}
\end{figure}

\noindent
\textbf{(1.b) Truthfulness: }Fig. \ref{fig3} shows the truthfulness of buyers and sellers in MIDA, where we select the winning buyer $d_2\in D_w^t$, losing buyer $d_7\notin D_w^t$, winning seller $s_7\in S_w^t$, and losing seller $s_1\notin S_w^t$ as examples to demonstrate the truthfulness of MIDA. Shown as (a) in Fig. \ref{fig3}, each $b^t_{7,j}$ changes from 0 to 1. At this time, other unit bids except $b^t_{7,j}$ are equal to their corresponding valuations. For the buyer $d_2$, its utility is $\hat{u}^t_2=0.695$ when giving the truthful bid. We can see that it cannot improve its utility by giving an untruthful bid $b^t_{2,j}$ where $s_j\in S$, which even reduces its utility to negative value. For the buyer $d_7$, its utility is $\hat{u}^t_7=0$ when giving the truthful bid. Shown as (b) in Fig. \ref{fig3}, it cannot improve its maximum utility by giving an untruthful bid as well. Shown as (c) (d) in Fig. \ref{fig3}, it is obvious that sellers obtain their maximum utilities by giving truthful asks. Therefore, the truthfulness of buyers and sellers can be held definitely.

\begin{figure}[!t]
	\centering
	\includegraphics[width=2.5in]{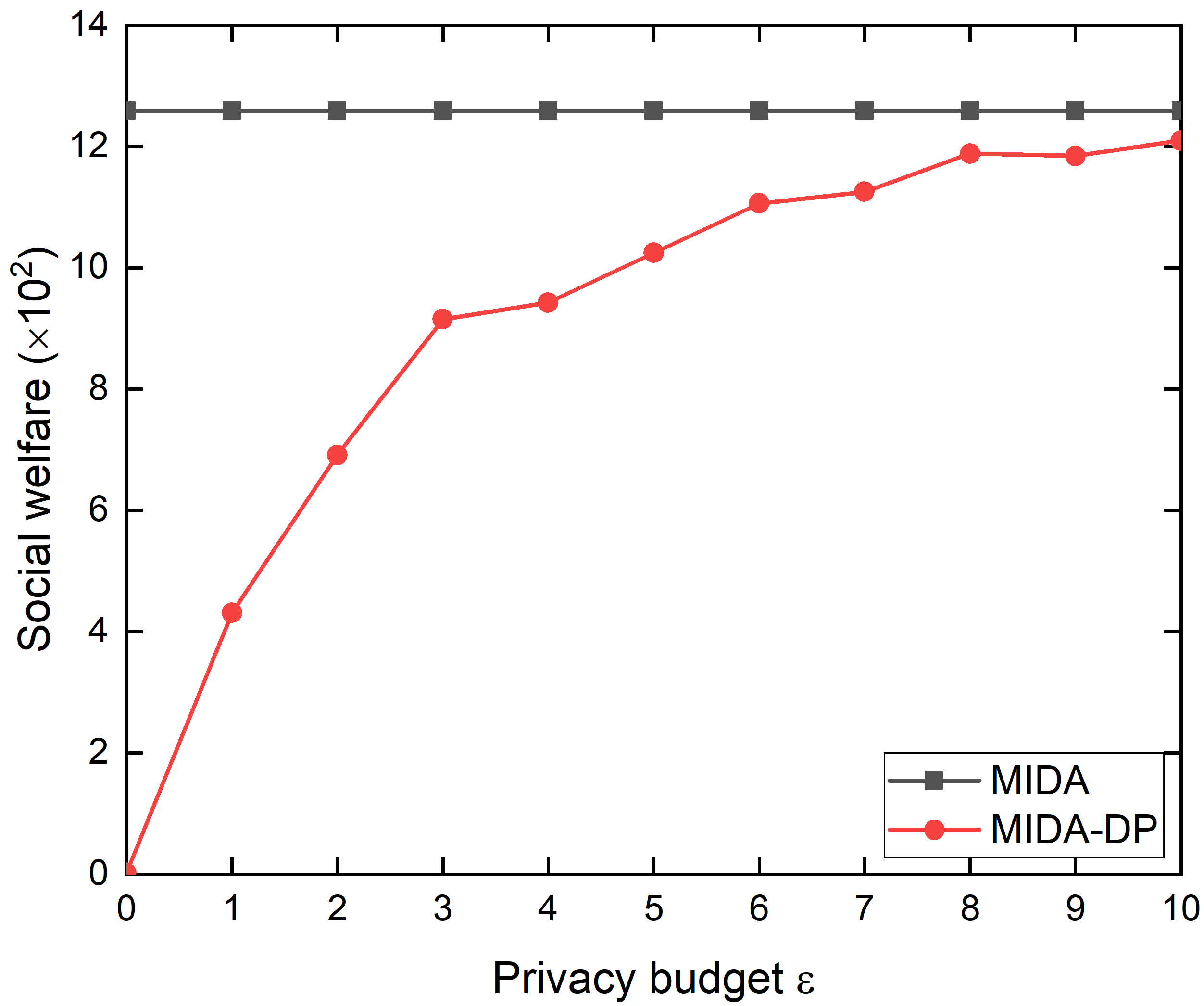}
	\centering
	\caption{The social welfare obtained by MIDA and MIDA-DP.}
	\label{fig4}
\end{figure}

\begin{figure}[!t]
	\centering
	\includegraphics[width=2.5in]{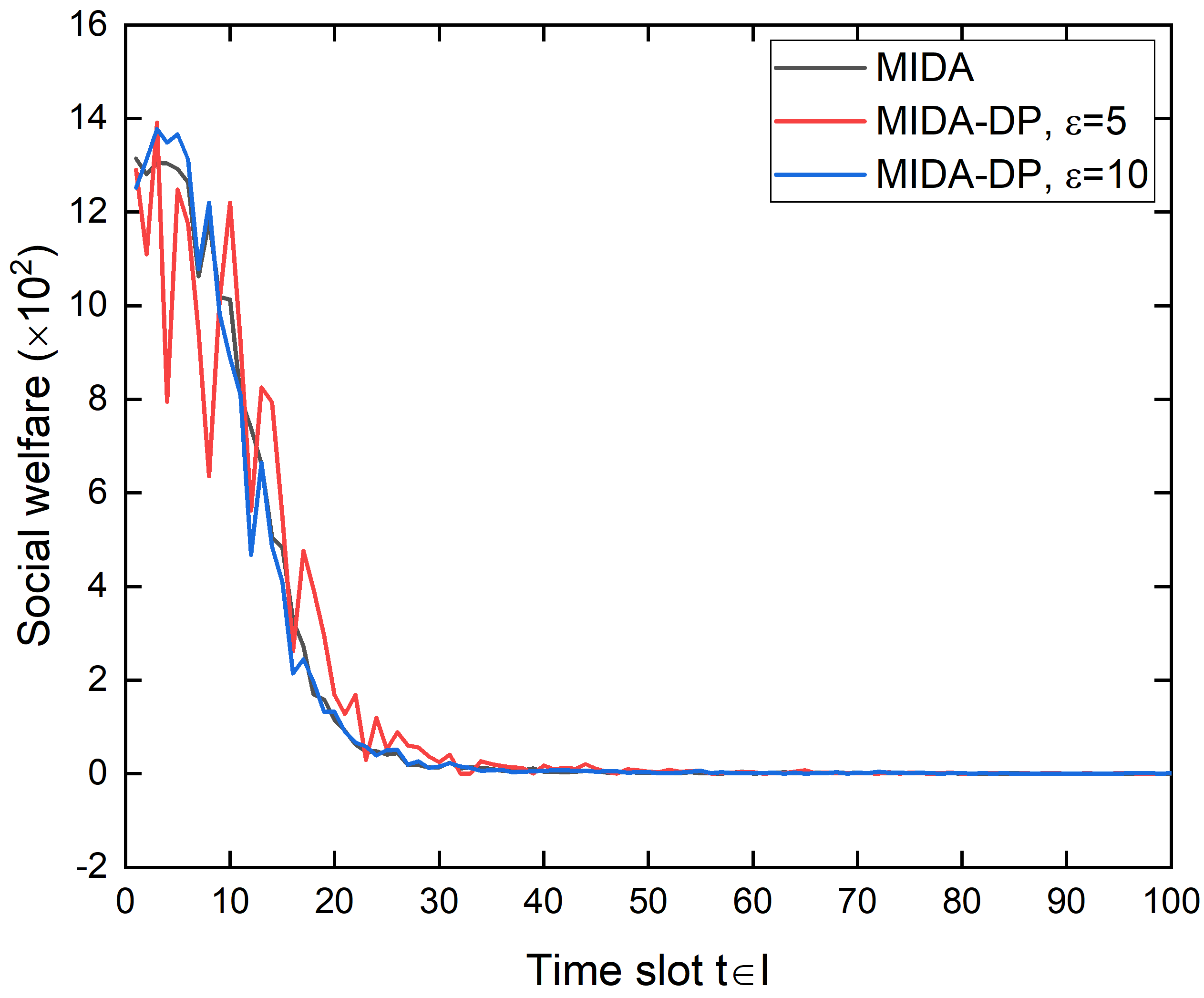}
	\centering
	\caption{The social welfare obtained by different online mechanisms.}
	\label{fig5}
\end{figure}

\textbf{Part 2: }We consider a static MIDA-DP mechanism in any time slot $t\in I$ based on the above settings, where we give $\alpha=1000$, $n=m=1000$, and $\gamma=50$. Here, the $r_i^t$ is distributed uniformly in interval $[0,10]$, $r_i^t\leq g_j^t$ for any $d_i\in D$ and $s_j\in S$, and the Constraint (8) has been satisfied. Then, we consider the online MIDA and MIDA-DP mechanisms in time interval $I$, where we give $\theta=30$. Thus, the Constraint (8) should be checked in each time slot. We use these two scenes to evaluate the performance of differential privacy and online mechanisms.

\noindent
\textbf{(2.a) Differential Privacy: }Fig. \ref{fig4} shows the social welfare obtained by MIDA and MIDA-DP with the different privacy budget $\varepsilon$, where we run the MIDA-DP mechansim $100$ times and take the average of them. Shown as Fig. \ref{fig4}, we can observe that the expected social welfare of MIDA-DP increases and approaches to the social welfare of MIDA (without differential privacy) gradually as its privacy budget increases. As we know, the larger the privacy budget is, the weaker the privacy protection is. Thus, we need to tradeoff the performance of social welfare and degree of privacy protection.

\noindent
\textbf{(2.b) Online Mechanisms: }Fig. \ref{fig5} shows the real-time social welfare obtain by online MIDA and MIDA-DP, where the time interval is $I=\{1,2,\cdots,100\}$. Shown as Fig. \ref{fig5}, the social welfare decreases gradually as the time slot increases. This is due to the constraint of $\theta$, and the permitted total computing resources in this time interval for some devices have been used up. Besides, we can see that the fluctuation of social welfare is larger when the privacy budget is smaller.

\textbf{Part 3: }We evaluate the MIDA-G mechanism where there are more than one device can be assigned to an edge server. Thus, we give that the $g^t_j$ is distributed in interval $[50,100]$. Similar to Part 2, we cansider a static MIDA-G-DP mechansim in any time slot $t\in I$ first, and then consider the online MIDA-G and MIDA-G-DP mechanisms in time interval $I$.

\begin{figure}[!t]
	\centering
	\includegraphics[width=2.5in]{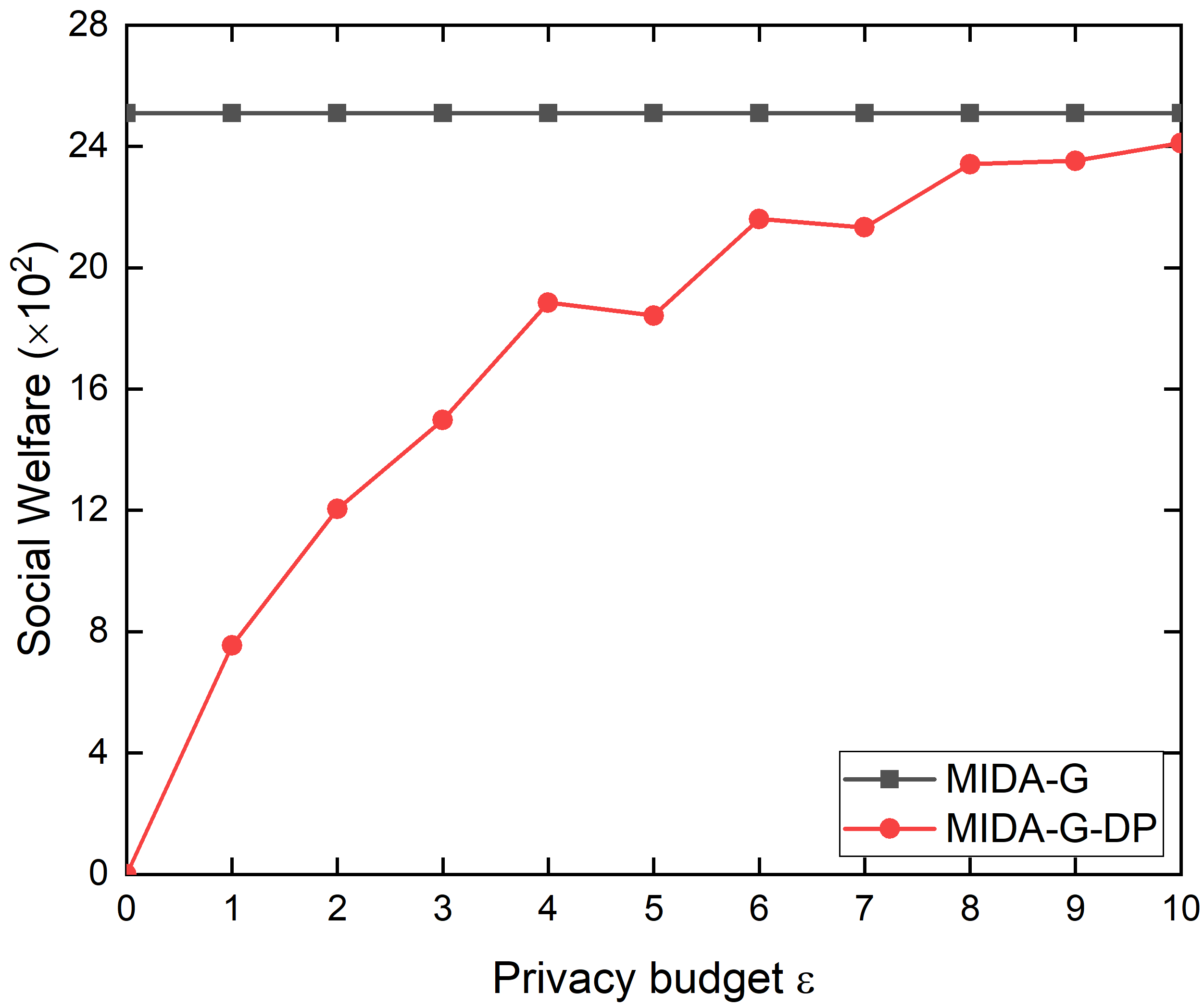}
	\centering
	\caption{The social welfare obtained by MIDA-G and MIDA-G-DP.}
	\label{fig6}
\end{figure}

\begin{figure}[!t]
	\centering
	\includegraphics[width=2.5in]{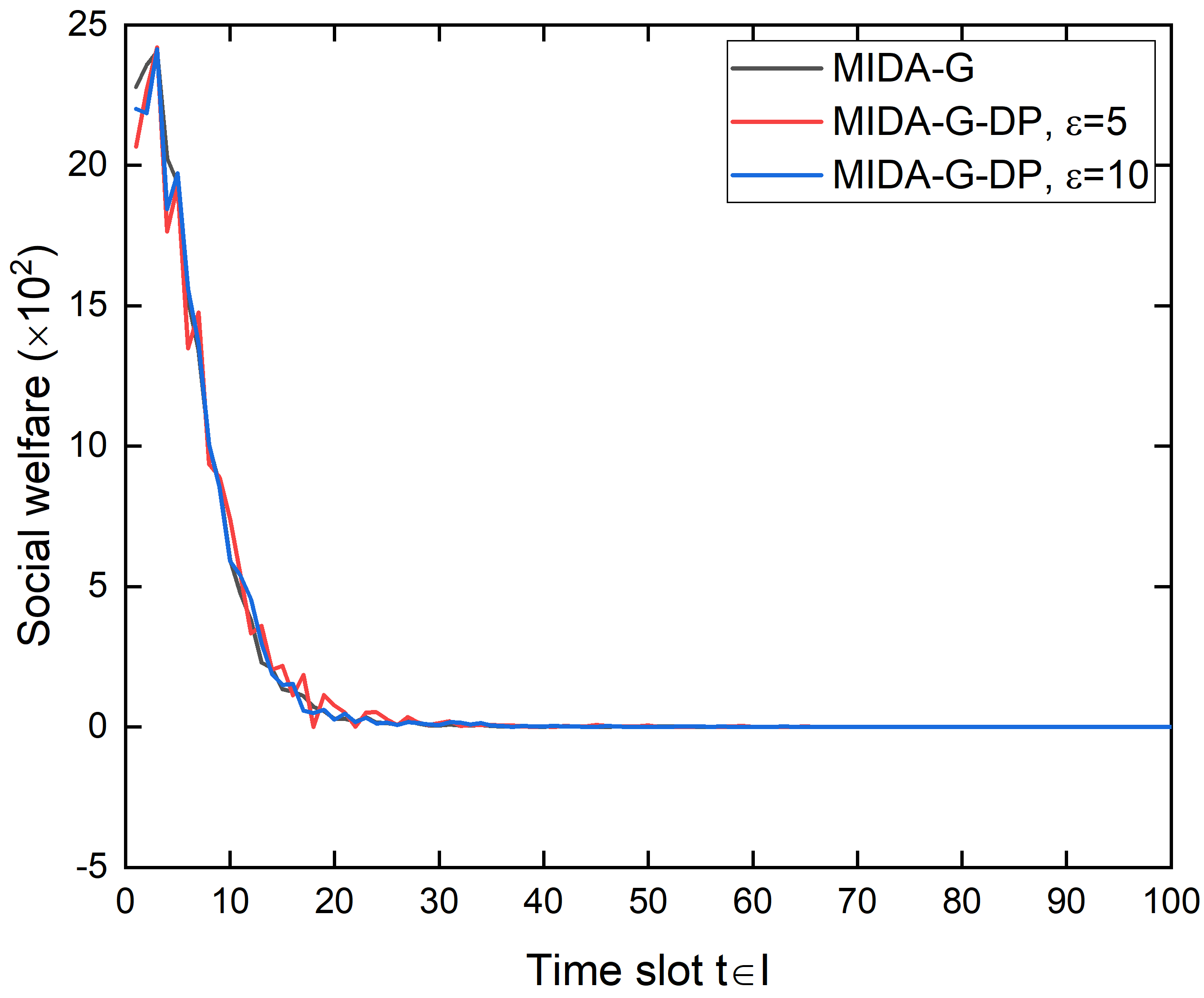}
	\centering
	\caption{The social welfare obtained by different online mechanisms.}
	\label{fig7}
\end{figure}

\noindent
\textbf{(3.a) Differential Privacy: }Fig. \ref{fig6} shows the social welfare obtained by MIDA-G and MIDA-G-DP with the different privacy budget $\varepsilon$, where we run the MIDP-G-DP mechanism $100$ times and take the average of them. According to our settings, we have $r_i^t<(1/5)\cdot g_j^t$ for any $d_i\in D$ and $s_j\in S$. Thus, there are at least five devices can be assigned to an edge server. Shown as Fig. \ref{fig6}, it shows a similar trend to MIDA-DP mechanism in Fig. \ref{fig4}. However, the social welfare of MIDA-G is improved significantly.

\noindent
\textbf{(3.b) Online Mechanisms: }Fig. \ref{fig7} shows the real-time social welfare obtained by online MIDA-G and MIDA-G-DP, where the time interval is $I=\{1,2,\cdots,100\}$. Shown as Fig. \ref{fig5}, it also shows a similar trend to MIDA and MIDA-DP in Fig. \ref{fig5}. Besides, the MIDA-G converges to zero more rapidly than the MIDA since much more devices can be assigned in earlier time slots, thus faster to reach the Constraint (8).

\section{Conclusion}
In this paper, we discussed a typical system model towards integrating blockchain and edge computing. First, we formulated the problem mathematically and modeled it by an online multi-item double auction mechanism. We designed a MIDA mechanism for a simplified special case and proved that it is individually rational, budget balanced, computationally efficient and truthful. Then, we analyzed system security and used the differential privacy to enhance security and privacy protection in order to prevent it from inference attack. Next, we proposed a MIDA-G mechanism changing our allocation from bijection to many-to-one, which is more general and realistic. Finally, we constructed a virtual scenario to test our mechanism by numerical simulations, which indicates the effectiveness and correctness of our auction algorithms and privacy protection.

\section*{Acknowledgment}

This work is partly supported by National Science Foundation under grant 1747818 and 1907472.

\ifCLASSOPTIONcaptionsoff
  \newpage
\fi

\bibliographystyle{IEEEtran}
\bibliography{references}

\begin{IEEEbiography}[{\includegraphics[width=1in,height=1.25in,clip,keepaspectratio]{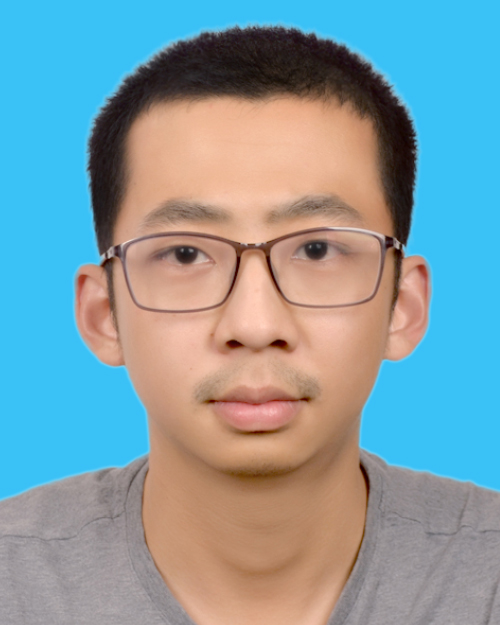}}]{Jianxiong Guo}
	is a Ph.D. candidate in the Department of Computer Science at the University of Texas at Dallas. He received his B.S. degree in Energy Engineering and Automation from South China University of Technology in 2015 and M.S. degree in Chemical Engineering from University of Pittsburgh in 2016. His research interests include social networks, data mining, IoT application, blockchain, and combinatorial optimization.
\end{IEEEbiography}

\begin{IEEEbiography}[{\includegraphics[width=1in,height=1.25in,clip,keepaspectratio]{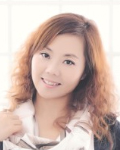}}]{Weili Wu}
	received the Ph.D. and M.S. degrees from the Department of Computer Science, University of Minnesota, Minneapolis, MN, USA, in 2002 and 1998, respectively. She is currently a Full Professor with the Department of Computer Science, The University of Texas at Dallas, Richardson, TX, USA. Her research mainly deals in the general research area of data communication and data management. Her research focuses on the design and analysis of algorithms for optimization problems that occur in wireless networking environments and various database systems.
\end{IEEEbiography}

\end{document}